\documentclass[a4,reqno,10pt]{amsart}

\usepackage{amsmath,amsthm,amssymb,xcolor,graphicx}
\usepackage{amscd,verbatim,a4wide}
\usepackage{float}
\usepackage{graphics,amsmath,amssymb}
\usepackage{amsthm}
\usepackage{amsfonts}
\usepackage{mathrsfs}
\usepackage{enumitem}
\usepackage[noadjust]{cite}
\usepackage{amsmath}
\allowdisplaybreaks[4]
\usepackage[all]{xy}
\usepackage{framed}
\colorlet{shadecolor}{black!15}

\usepackage{marginnote}

\usepackage[hidelinks]{hyperref}
\hypersetup{colorlinks=false}

\usepackage{listings}

\numberwithin{equation}{section}

\theoremstyle{plain}
\newtheorem{theorem}{Theorem}
\newtheorem{lemma}[theorem]{Lemma}

\theoremstyle{definition}

\theoremstyle{remark}

\newcommand{\LHS}{\operatorname{LHS}}

\newcommand{\func}[2]{#1 \left( #2 \right)}

\newcommand{\fg}{{\mathfrak g}}

\newcommand{\fsl}{{\mathfrak{sl}}}

\begin{document}

\sloppy

\title[$q$-identities for parafermion theories]{$q$-identities for parafermion theories}

\author[Peter Bouwknegt]{Peter Bouwknegt}
\address[Peter Bouwknegt]{
Mathematical Sciences Institute,
Australian National University, 
Canberra, ACT 2601, Australia}
\email{peter.bouwknegt@anu.edu.au}

\author[Shane Chern]{Shane Chern}
\address[Shane Chern]{
Department of Mathematics and Statistics,
Dalhousie University, 
Halifax, NS, B3H 4R2, Canada}
\email{chenxiaohang92@gmail.com, xh375529@dal.ca}

\author[Bolin Han]{Bolin Han}
\address[Bolin Han]{
Mathematical Sciences Institute,
Australian National University, 
Canberra, ACT 2601, Australia}
\email{bolin.han@anu.edu.au}

\begin{abstract}
In this paper, we will prove a series of $q$-identities suggested by the realisation of certain conformal field theories by so-called `coupled free fermions.'  We will consider $q$-series arising from coupled free fermions constructed by the parafermion coset construction as well as from scaled root lattices, and some interesting relations between the two.
\end{abstract}


\maketitle

\section{Introduction}

It has been demonstrated in the literature that combining number theory and mathematical physics techniques can produce useful and interesting, or even unexpected results for both areas. Especially, $q$-identities have played important roles in conformal field theories and vice versa. For instance, we have the famous Jacobi triple product directly related to the Weyl denominator formula, the Hecke type identities motivated by the string functions of affine Lie algebra developed in \cite{Kac1980} and then proved by the classical theory of $q$-series in \cite{And1984}, and more recently some generalised Rogers--Ramanujan identities derived from a theory of parafermions of the type $(\widehat{\fsl_{n+1}})_2/\widehat{\mathfrak{u}(1)}^n$ in \cite{Bel2013}. 

In this paper, we will reveal some curious 
dentities, 
including some of Rogers--Ramanujan type, observed from universal chiral partition functions (`fermionic partition functions') of parafermions of the type $(\widehat{\fsl_{n+1}})_2/\widehat{\mathfrak{u}(1)}^n$ and characters of parafermions constructed from scaled root lattices of the type $A_n/\sqrt{2}$. They are then rigorously proved using various techniques for basic hypergeometric series. With these identities, we are also able to derive explicit expressions of some string functions purely in terms of $\eta$-functions, which can be hard to compute even with the machinery given in \cite{Kac1984}, in a relatively straightforward manner. The success of proving these specialised $q$-identities further motivates us to investigate more general $q$-identities. This work once again illustrates the vibrant collaboration between number theory and mathematical physics.

The physical motivation behind the $q$-series arising in this paper is discussed in detail in \cite{Han23, BH24}. Briefly speaking, one can generalise the idea of (uncoupled) free fermions to coupled free fermions and study the exclusion statistics of coupled free fermions in conformal field theories by constructing particular bases of their representations, which lead to computation of their characters in terms of universal chiral partition functions.

\section{String functions for level-2 $\fsl_n$ parafermions}\label{sec:strfun}

\subsection{String functions and parafermionic characters}

Let $\fg$ be a finite-dimensional Lie algebra. Denote the (untwisted) affine Kac--Moody algebra at level $k$, associated to $\fg$, by $\widehat{\mathfrak{g}}_k$. Let $\Lambda$ be an integrable dominant weight of $\widehat{\mathfrak{g}}_k$, and $\lambda$ be a maximal weight in the representation with highest weight $\Lambda$. Then the so-called \emph{string functions} $c^{\Lambda}_{\lambda}$ are defined in \cite{Kac1980} by Kac and Peterson: 
\[c^{\Lambda}_{\lambda}(\tau)=q^{s_{\Lambda}(\lambda)}\sum_{n\ge0}\operatorname{mult}_{\Lambda}(\lambda-n\delta)e^{2\pi i\tau},\]
where 
\[s_{\Lambda}(\lambda)=\dfrac{|\Lambda+\rho|^2}{2(k+h^{\vee})}-\dfrac{|\rho|^2}{2h^{\vee}}-\dfrac{|\lambda|^2}{2k},\]
with $\rho$ being half of the sum of all positive roots of $\hat{\fg}$, and $h^{\vee}$ the dual Coxeter number of $\hat{\fg}$. The modular transformation law for string functions is provided in \cite{Kac1980} as well.

It is also known that given $\widehat{\mathfrak{g}}_k$, one can construct a \emph{parafermionic conformal field theory} as described in \cite{Gep1987}. Algebraically, this theory can be thought of the Goddard--Kent--Olive (GKO) coset construction of the (generalised) Wess--Zumino--Witten (WZW) model (see, for example, \cite{Fra1997}) based on $\widehat{\mathfrak{g}}_k$ constrained by $\widehat{\mathfrak{u}(1)}^{n}$, where $n=\operatorname{rank}(\mathfrak{g})$; see \cite{Gep1989}. We denote such a construction as $\operatorname{PF}_k(\fg)$. Explicitly, we define the parafermions $\psi_{\alpha}$ for every $\alpha\in Q$, where $Q$ is the root lattice of $\mathfrak{g}$. Meanwhile, we identify $\psi_{\alpha}=\psi_{\beta}$ if $\alpha\equiv\beta\bmod kQ$. The operator product expansions (OPEs) of these parafermions are 
\begin{align}
    \psi_{\alpha}(z)\psi_{-\alpha}(w)\sim\dfrac{1}{(z-w)^{2-|\alpha|^2/k}},\label{eq:paraOPE-1}\\
     \psi_{\alpha}(z)\psi_{\beta}(w)\sim\dfrac{c_{\alpha,\beta}\psi_{\alpha+\beta}(w)}{(z-w)^{1+(\alpha,\beta)/k}},\label{eq:paraOPE-2}
\end{align}
where $c_{\alpha,\beta}$ are some constants to be fixed.

We notice that when $k=2$ and $\fg$ is simply-laced, this construction is of particular interest because then the OPEs \eqref{eq:paraOPE-1} and \eqref{eq:paraOPE-2} will reduce to
\begin{align}
    &\psi_{\alpha}(z)\psi_{\alpha}(w)\sim\dfrac{1}{(z-w)},\\
     &\psi_{\alpha}(z)\psi_{\beta}(w)\sim\dfrac{c_{\alpha,\beta}\psi_{\alpha+\beta}(w)}{(z-w)^{1/2}},
\end{align}
if $\alpha+\beta\bmod 2Q$ is a root of $\fg$, so we have real free fermions $\psi_{\alpha}$ that couple to each other in a certain way. We may call such a system a ``\emph{coupled free fermion system}.'' 

The modules of $\operatorname{PF}_k(\fg)$ are again characterised by $\Lambda$ and $\lambda$, denoted by $L^{\Lambda}_{\lambda}$. The modules $L^{\Lambda}_{\lambda}$ and $L^{\Lambda'}_{\lambda'}$ are equivalent if $$\lambda'\equiv \lambda\pmod{kQ}$$ 
or $$\Lambda'=\sigma\Lambda\;\;\text{and}\;\;\lambda'=\sigma\lambda,$$
where $\sigma$ is an automorphism of the Dynkin diagram of $\hat{\fg}$. The \emph{characters} (or \emph{partition functions}) of the modules are given by
\[b^{\Lambda}_{\lambda}(\tau):=\operatorname{Tr}e^{\left(L_0-c/24\right)2\pi i\tau}\]
restricting to the module $L^{\Lambda}_{\lambda}$, where $L_0$ is the operator that diagonalises the fock space of $\operatorname{PF}_k(\fg)$. It was proven in \cite{Gep1987} that the parafermionic characters $b^{\Lambda}_{\lambda}$ of $\operatorname{PF}_k(\fg)$ are simply related to the 
string functions $c^{\Lambda}_{\lambda}$ by
\[b^{\Lambda}_{\lambda}(\tau)=\eta(\tau)^{n} c^{\Lambda}_{\lambda}(\tau),\]
where $\eta(\tau)$ is the Dedekind eta function.

Given such a close relation between string functions of $\widehat{\fg}_k$ and characters of $\operatorname{PF}_k(\fg)$, we may study them all at once and take advantage of either side as needed.

In the rest of this section, we list a few examples of string functions and parafermionic characters for later reference in this paper. In these examples, we write $q=e^{2\pi i\tau}$ by convention and use Dynkin labels to represent the weights.

\subsection{$\fsl_2$}
The string functions of $(\widehat{\mathfrak{sl}_2})_2$ have been explicitly computed in \cite{Kac1984} (Example 3 in Section 4.6). 
Hence we have explicit expressions for all characters $b^{\Lambda}_{\lambda}$ of $\operatorname{PF}_2(\mathfrak{sl}_2)$ as follows:
\begin{align*}
    b^{20}_{20}(\tau)&=\frac{1}{2}\big(\eta(2\tau)^{-1}\eta(\tau/2)^{-1}\eta(\tau)^2+\eta(\tau)^{-1}\eta(\tau/2)\big),\\
    b^{20}_{02}(\tau)&=\frac{1}{2}\big(\eta(2\tau)^{-1}\eta(\tau/2)^{-1}\eta(\tau)^2-\eta(\tau)^{-1}\eta(\tau/2)\big),\\
    b^{11}_{11}(\tau)&=\eta(\tau)^{-1}\eta(2\tau),
\end{align*}

\subsection{$\fsl_3$}

The string functions of $(\widehat{\mathfrak{sl}_3})_2$ have been explicitly computed in \cite{Kac1984} 
(Example 4 in Section 4.6). Hence we have explicit expressions for all characters $b^{\Lambda}_{\lambda}$ of $\operatorname{PF}_2(\mathfrak{sl}_3)$. 
We claim that they can be rewritten as follows:
\begin{subequations}
\begin{align}
   b_{200}^{200}(\tau) &=\eta(\tau)^{-2}q^{\frac{1}{30}}\big((q^\frac{1}{2};q^\frac{1}{2})_\infty(q,q^{\frac{3}{2}},q^{\frac{5}{2}};q^{\frac{5}{2}})_\infty+ q^{\frac{1}{2}}(q^2;q^2)_\infty(q^2,q^{8},q^{10};q^{10})_\infty\big),\label{eq:b200200}\\
	b_{011}^{200}(\tau)&=\eta(\tau)^{-2}q^{\frac{8}{15}}(q^2;q^2)_\infty(q^2,q^{8},q^{10};q^{10})_\infty,\label{eq:b200011}\\
	b_{110}^{110}(\tau)&=\eta(\tau)^{-2}q^{\frac{2}{15}}(q^2;q^2)_\infty(q^4,q^{6},q^{10};q^{10})_\infty,\label{eq:b110110}\\
	b_{002}^{110} (\tau)&=\eta(\tau)^{-2}q^{\frac{2}{15}}\big((q^2;q^2)_\infty(q^4,q^{6},q^{10};q^{10})_\infty-(q^\frac{1}{2};q^\frac{1}{2})_\infty(q^{\frac{1}{2}},q^{2},q^{\frac{5}{2}};q^{\frac{5}{2}})_\infty\big),\label{eq:b110002}
\end{align}
\end{subequations}
where the \emph{$q$-Pochhammer symbols} are defined for $n\in\mathbb{N}\cup\{\infty\}$:
\begin{align*}
	(a;q)_n&:=\prod_{j=0}^{n-1} (1-aq^{j} ),\\
	(a_1, a_2, \ldots, a_r;q)_{n} &:= (a_1;q)_n (a_2;q)_n \cdots (a_r;q)_n.
\end{align*}
Note that for notational convenience, we often write hereinafter, with $E$ standing for Euler, that
\begin{align*}
E(q):=(q;q)_\infty.
\end{align*}
It is worth pointing out that the above four expressions will be examined as part of the proof of Theorem \ref{th:sl3UCPF}.

\subsection{$\fsl_4$} 
We have found very nice explicit expressions for string functions of $(\widehat{\mathfrak{sl}_4})_2$, and hence also characters of $\operatorname{PF}_2(\mathfrak{sl}_4)$, purely in terms of $\eta$-functions.
\begin{theorem}\label{th:sl4strfun}
Let
\begin{equation*}
    \arraycolsep=1.4pt\def\arraystretch{2.2}
    \begin{array}{lp{1.6cm}l}
         \multicolumn{3}{l}{b^{2000}_{2000}(\tau)+b^{2000}_{0020}(\tau)=\dfrac{\eta(4\tau)^4\eta(6\tau)^8}{\eta(\tau)\eta(2\tau)^4\eta(3\tau)^3\eta(12\tau)^4} +\dfrac{\eta(2\tau)^8\eta(3\tau)\eta(12\tau)^4}{\eta(\tau)^5\eta(4\tau)^4\eta(6\tau)^4},}\\
        b^{2000}_{2000}(\tau)-b^{2000}_{0020}(\tau) = \dfrac{\eta(\tau)^2}{\eta(2\tau)^2}, & & b^{2000}_{0101}(\tau) = \dfrac{\eta(2\tau)^2\eta(6\tau)^2}{\eta(\tau)^3\eta(3\tau)},\\
        b^{0101}_{2000}(\tau) = b^{0101}_{0002}(\tau) = \dfrac{3\eta(6\tau)^3}{\eta(\tau)^2\eta(2\tau)}, & & b^{0101}_{0101}(\tau) = \dfrac{\eta(2\tau)^3\eta(3\tau)^2}{\eta(\tau)^4\eta(6\tau)},\\
        b^{1100}_{1100}(\tau)=\dfrac{\eta(4\tau)^5}{\eta(\tau)^3\eta(8\tau)^2}, & & b^{1100}_{0011}(\tau)=\dfrac{2\eta(2\tau)^2\eta(8\tau)^2}{\eta(\tau)^3\eta(4\tau)}.
    \end{array}
\end{equation*}
Then each function $\eta(\tau)^{-3}b^{\Lambda}_{\lambda}$ with $(\Lambda,\lambda)$ appearing above is an expression for the corresponding string function $c^{\Lambda}_{\lambda}$ of $(\widehat{\mathfrak{sl}_4})_2$. Hence each $b^{\Lambda}_{\lambda}$ is an expression 
for the character of the module $L^{\Lambda}_{\lambda}$ of $\operatorname{PF}_2(\mathfrak{sl}_4)$.
\end{theorem}

\begin{proof}
Given that these expressions for $b^{\Lambda}_{\lambda}$ are all in terms of the Dedekind eta function, one can easily compute their modular transformations under $\tau\mapsto \tau+1$ and $\tau\mapsto -1/\tau$ and compare the results computed using Theorem A(1) in \cite{Kac1984} as follows:
\begin{align*}
    b^{2000}_{2000}\left(-\frac{1}{\tau}\right)-b^{2000}_{0020}\left(-\frac{1}{\tau}\right)&=2\left(b^{1100}_{1100}(\tau)+b^{1100}_{0011}(\tau)\right),\\
    b^{2000}_{2000}\left(-\frac{1}{\tau}\right)+b^{2000}_{0020}\left(-\frac{1}{\tau}\right)&=\dfrac{1}{2\sqrt{3}}\left(b^{2000}_{2000}(\tau)+b^{2000}_{0020}(\tau) + 6 b^{2000}_{0101}(\tau) + 6b^{0101}_{0101}(\tau) +  2b^{0101}_{2000}(\tau)\right),\\
    b^{0101}_{2000}\left(-\frac{1}{\tau}\right)=b^{0101}_{0002}\left(-\frac{1}{\tau}\right)&=\dfrac{1}{2\sqrt{3}}\left(b^{2000}_{2000}(\tau)+b^{2000}_{0020}(\tau) + 6 b^{2000}_{0101}(\tau) -3 b^{0101}_{0101}(\tau) -  b^{0101}_{2000}(\tau)\right),\\
    b^{2000}_{0101}\left(-\frac{1}{\tau}\right)&=\dfrac{1}{4\sqrt{3}}\left(b^{2000}_{2000}(\tau)+b^{2000}_{0020}(\tau) - 2 b^{2000}_{0101}(\tau) -2 b^{0101}_{0101}(\tau) + 2 b^{0101}_{2000}(\tau)\right),\\
    b^{0101}_{0101}\left(-\frac{1}{\tau}\right)&=\dfrac{1}{2\sqrt{3}}\left(b^{2000}_{2000}(\tau)+b^{2000}_{0020}(\tau) - 2 b^{2000}_{0101}(\tau) + b^{0101}_{0101}(\tau) - b^{0101}_{2000}(\tau)\right),\\
    b^{1100}_{1100}\left(-\frac{1}{\tau}\right)&=\dfrac{1}{4}\left(b^{2000}_{2000}(\tau)-b^{2000}_{0020}(\tau) + 2 b^{1100}_{1100}(\tau) -2 b^{1100}_{0011}(\tau)\right),\\
     b^{1100}_{0011}\left(-\frac{1}{\tau}\right)&=\dfrac{1}{4}\left(b^{2000}_{2000}(\tau)-b^{2000}_{0020}(\tau) - 2 b^{1100}_{1100}(\tau) +2 b^{1100}_{0011}(\tau)\right).
\end{align*}
Also, using Theorem A(4)\footnote{With a closer look at the proof of Theorem A(3) in \cite{Kac1984}, we believe that Theorem A(4), op.~cit., holds true if we replace all the groups $\Gamma(Nm)$ and $\Gamma(N(m+g))$ by $\Gamma_0(Nm)$ and $\Gamma_0(N(m+g))$, 
respectively.} in \cite{Kac1984} and Sturm's bound for modular forms (see, for example, Theorem 6.4.7 in \cite{MDG2015}), we know that comparing 
coefficients of $\eta(\tau)^{-3}b^{\Lambda}_{\lambda}$ with the numerical result for $ c^{\Lambda}_{\lambda}$, computed from its definition, 
up to the order of $q^{100}$ will be sufficient to conclude that $\eta(\tau)^{-3}b^{\Lambda}_{\lambda}=c^{\Lambda}_{\lambda}$. 
All of the above have been verified with the help of \textit{SageMath} and \textit{Mathematica}.
\end{proof}
In fact, these expressions for $b^{\Lambda}_{\lambda}$ arise in the way of proving Theorem \ref{th:sl4UCPF} and \ref{th:A2lat}, with the related results summarised as \eqref{eq:VOA-1-1}, \eqref{eq:VOA-1-2} and \eqref{eq:VOA-2}.

\section{Universal Chiral partition functions for level-2 $\fsl_n$ parafermions} \label{sec:Level-2}

\subsection{Universal chiral partition functions}
It can be seen from \cite{Gep1987} that the representation theory of $\operatorname{PF}_2(\mathfrak{sl}_n)$ involves intertwiners between various representations, 
so it is natural to think of fractional exclusion statistics where the statistical interactions of intertwiners are encoded into a matrix $\mathbf{G}$; see \cite{Hal1991}. 

The so-called \emph{universal chiral partition function} (UCPF) \cite{Bou2000} is in the form of
\[\func{\operatorname{UCPF}}{\mathbf{G};\mathbf{a}}:=q^{\delta}\sum_{N_i\ge 0}\dfrac{q^{\frac{1}{2}\mathbf{N}^{\mathsf{T}}\cdot\mathbf{G}\cdot\mathbf{N}-\mathbf{a}\cdot \mathbf{N}}}{\prod_i (q;q)_{N_i}}\]
where $q^{\delta}$ is a prefactor, $\mathbf{N}=(N_1,N_2,\dots,N_m)^{\mathsf{T}}$, $\mathbf{a}=(a_1,a_2,\dots,a_m)$ and $m$ is the number of quasi-particles.

In the case of $\operatorname{PF}_2(\mathfrak{sl}_n)$, generalised commutation relations \cite{Noy2007,Han23,BH24} allow us to analyse the exclusion statistics of intertwiners. 
This motivates us to conjecture a statistical interaction matrix $\mathbf{G}$ for $\operatorname{PF}_2(\mathfrak{sl}_n)$. We will examine a few examples to see the connection between 
UCPFs and parafermionic characters $b^{\Lambda}_{\lambda}$.


\subsection{$\fsl_2$} In this case the exclusion statistic is trivial, so we define
$$\mathbf{G}_2:=\frac{1}{2}\big(2\big),$$
where $\big(2\big)$ is understood as a one-by-one matrix. We have relations as follows.
\begin{theorem}\label{th:sl2UCPF}
\begin{subequations}
\begin{align}
    \func{\operatorname{UCPF}}{\mathbf{G}_2;\mathbf{0}}&=q^{-\frac{1}{48}}\sum_{N\ge 0}\dfrac{q^{\frac{1}{2}N^2}}{(q;q)_N}=b^{20}_{20}+b^{20}_{02},\label{eq:sl2untwi}\\
    \func{\operatorname{UCPF}}{\mathbf{G}_2;
 (\tfrac{1}{2})}&=q^{\frac{1}{16}-\frac{1}{48}}\sum_{N\ge 0}\dfrac{q^{\frac{1}{2}(N^2-N)}}{(q;q)_N}=2b^{11}_{11}\label{eq:sl2twi}.
\end{align}
\end{subequations}
\end{theorem}

\begin{proof}
The two relations are instances of Euler's second summation formula \eqref{eq:Eul-2}. For \eqref{eq:sl2untwi}, 
we use \eqref{eq:Eul-2} with $z=-q^{\frac{1}{2}}$. For \eqref{eq:sl2twi}, we again use \eqref{eq:Eul-2} but with $z=-1$.
\end{proof}
\subsection{$\fsl_3$}
According to the explicit analysis of $\operatorname{PF}_2(\mathfrak{sl}_3)$ in \cite{Ard2001,Han23,BH24}, we first conjecture that
\begin{align}\label{eq:G3-def}
\mathbf{G}_3:=\dfrac{1}{2}\begin{pmatrix}
2&1&1\\
1&2&1\\
1&1&2
\end{pmatrix}.
\end{align}
Then noticing that there are (up to equivalence) one untwisted and one twisted representation of $\operatorname{PF}_2(\mathfrak{sl}_3)$, we consider
\begin{align*}
    \func{\operatorname{UCPF}}{\mathbf{G}_3;\mathbf{0}}&=q^{-\frac{1}{20}}\sum_{N_i\ge 0}\dfrac{q^{\frac{1}{2}\mathbf{N}^{\mathsf{T}}\cdot\mathbf{G}_3\cdot\mathbf{N}}}{\prod_i (q;q)_{N_i}},\\
  \func{\operatorname{UCPF}}{\mathbf{G}_3;\tfrac{1}{2} (0,1,1)}&=q^{\frac{1}{10}-\frac{1}{20}}\sum_{N_i\ge 0}\dfrac{q^{\frac{1}{2}\mathbf{N}^{\mathsf{T}}\cdot\mathbf{G}_3\cdot\mathbf{N}-\frac{1}{2}(N_2+N_3)}}{\prod_i (q;q)_{N_i}},
\end{align*}
where
\begin{align}\label{eq:G3-N}
    \mathbf{N}^{\mathsf{T}}=(N_1,N_2,N_3).
\end{align}

It is easy to observe relations between these UCPFs and the parafermionic characters $b^{\Lambda}_{\lambda}$ numerically. In particular, we have the following result. 
\begin{theorem}\label{th:sl3UCPF}
\begin{subequations}
\begin{align}
    \func{\operatorname{UCPF}}{\mathbf{G}_3; \mathbf{0}}&=b_{200}^{200}+3b_{011}^{200},\label{eq:sl3untwi}\\
    \func{\operatorname{UCPF}}{\mathbf{G}_3;\tfrac{1}{2} (0,1,1)}&=b_{002}^{110}+3b_{110}^{110}.\label{eq:sl3twi}
\end{align}
\end{subequations}
\end{theorem}
\begin{proof}
We begin by noticing that \eqref{eq:sl3untwi} and \eqref{eq:sl3twi} are, respectively, equivalent to
\begin{subequations}
\begin{align}
	\sum_{N_i\ge 0} \frac{q^{\mathbf{N}^{\mathsf{T}}\cdot\mathbf{G}_3\cdot\mathbf{N}}}{\prod_i(q^2;q^2)_{N_i}}&=\frac{(q;q)_\infty(q^2,q^{3},q^{5};q^{5})_\infty}{(q^2;q^2)^2_\infty}+\frac{4q(q^4;q^4)_\infty(q^4,q^{16},q^{20};q^{20})_\infty}{(q^2;q^2)^2_\infty},\label{eq:triple-AG-1-2}\\
	\sum_{N_i\ge 0} \frac{q^{\mathbf{N}^{\mathsf{T}}\cdot\mathbf{G}_3\cdot\mathbf{N}-N_2-N_3}}{\prod_i(q^2;q^2)_{N_i}}&=-\frac{(q;q)_\infty(q,q^{4},q^{5};q^{5})_\infty}{(q^2;q^2)^2_\infty}+\frac{4(q^4;q^4)_\infty(q^8,q^{12},q^{20};q^{20})_\infty}{(q^2;q^2)^2_\infty}.\label{eq:triple-AG-2-2}
\end{align}
\end{subequations}
For \eqref{eq:triple-AG-1-2}, we first show that
\begin{align}
    \LHS\eqref{eq:triple-AG-1-2}=\frac{(q^2;q^2)_\infty(q^2,q^8,q^{10};q^{10})_\infty(q^6,q^{10},q^{14};q^{20})_\infty}{(q;q)_\infty(q^4;q^4)_\infty(q^3,q^5,q^{7};q^{10})_\infty}+\frac{2q(q^4;q^4)_\infty(q^4,q^{16},q^{20};q^{20})_\infty}{(q^2;q^2)^2_\infty},\label{eq:triple-AG-1}
\end{align}
where the right-hand side is adopted from \cite{Kac1984}.
To see this, we take $(z_1,z_2,z_3)=(1,1,1)$ in \eqref{eq:G3-general}. Thus,
\begin{align*}
    \LHS\eqref{eq:triple-AG-1-2}&=(-q;q^2)_\infty \sum_{M\ge 0}\frac{q^{3M^2}}{(q;q^2)_M (q^4;q^4)_{M}}+ (-q^2;q^2)_\infty \sum_{M\ge 0} \frac{2 q^{3M^2+3M+1} }{(q^2;q^2)_{M} (q^2;q^4)_{M+1}}.
\end{align*}
Now, using \eqref{eq:S19} with $q\mapsto -q$ and \eqref{eq:S44} with $q\mapsto q^2$, we arrive at \eqref{eq:triple-AG-1}. We further rewrite \eqref{eq:triple-AG-1} with the Rogers--Ramanujan functions \eqref{eq:G-def} and \eqref{eq:H-def}:
	\begin{align*}
	\LHS\eqref{eq:triple-AG-1-2}=\frac{E(q^2)^4}{E(q)^2E(q^4)^2}G(-q)+\frac{2qE(q^4)^2}{E(q^2)^2}H(q^4).
	\end{align*}
	To show \eqref{eq:triple-AG-1-2}, it suffices to prove
	\begin{align}\label{eq:G-modular-eqn}
	\frac{E(q^2)^4}{E(q)^2E(q^4)^2}G(-q)G(q^4)-\frac{E(q)^2}{E(q^2)^2}G(q)G(q^4) = \frac{2qE(q^4)E(q^{20})}{E(q^2)^2}.
	\end{align}
	By \eqref{eq:GH+} and \eqref{eq:GH-}, we may obtain an expression for $G(q)G(q^4)$ and thus also for $G(-q)G(q^4)$. Then,
	\begin{align*}
	\LHS\eqref{eq:G-modular-eqn}
	&=\frac{E(q^2)^3E(q^5)^2}{2E(q)^2E(q^4)^2E(q^{10})}-\frac{E(q)^2E(q^{10})^5}{2E(q^2)^3E(q^5)^2E(q^{20})^2}\\
	&=\frac{\phi(q)\phi(-q^5)-\phi(-q)\phi(q^5)}{2E(q^2)^2},
	\end{align*}
	thereby confirming \eqref{eq:G-modular-eqn} via the following modular equation of degree $5$ (see \cite[p.~278]{Ber1991}):
\begin{align}\label{eq:modular-deg-5}
	\phi(q)\phi(-q^5)-\phi(-q)\phi(q^5) = 4q E(q^4)E(q^{20}).
\end{align}
Here, Ramanujan's theta function $\phi(q)$ is defined in \eqref{eq:phi-def}.

For \eqref{eq:triple-AG-2-2}, we first establish that
\begin{align}
    \LHS\eqref{eq:triple-AG-2-2}=\frac{(q^2;q^2)_\infty(q^4,q^6,q^{10};q^{10})_\infty(q^2,q^{10},q^{18};q^{20})_\infty}{(q;q)_\infty(q^4;q^4)_\infty(q,q^5,q^{9};q^{10})_\infty}
    +\frac{2(q^4;q^4)_\infty(q^8,q^{12},q^{20};q^{20})_\infty}{(q^2;q^2)^2_\infty},\label{eq:triple-AG-2}
\end{align}
where the right-hand side is also adopted from \cite{Kac1984}.
To see this, we take $(z_1,z_2,z_3)=(1,q^{-1},q^{-1})$ in \eqref{eq:G3-general}. Thus,
\begin{align*}
    \LHS\eqref{eq:triple-AG-2-2}&=(-q;q^2)_\infty \sum_{M\ge 0}\frac{q^{3M^2-2M}}{(q;q^2)_M (q^4;q^4)_{M}}+ 
    (-q^2;q^2)_\infty \sum_{M\ge 0} \frac{2 q^{3M^2+M} }{(q^2;q^2)_{M} (q^2;q^4)_{M+1}}.
\end{align*}
We then make use of \eqref{eq:S15} with $q\mapsto -q$ and \eqref{eq:S46} with $q\mapsto q^2$ to derive \eqref{eq:triple-AG-2}. 
We further rewrite \eqref{eq:triple-AG-2} in terms of the Rogers--Ramanujan functions:
	\begin{align*}
	\LHS\eqref{eq:triple-AG-2-2}=\frac{E(q^2)^4}{E(q)^2E(q^4)^2}H(-q)+\frac{2E(q^4)^2}{E(q^2)^2}G(q^4).
	\end{align*}
	So \eqref{eq:triple-AG-2-2} is equivalent to
	\begin{align}\label{eq:H-modular-eqn}
	\frac{E(q^2)^4}{E(q)^2E(q^4)^2}H(-q)H(q^4)+\frac{E(q)^2}{E(q^2)^2}H(q)H(q^4) = \frac{2E(q^4)E(q^{20})}{E(q^2)^2}.
	\end{align}
	We also deduce from \eqref{eq:GH+} and \eqref{eq:GH-} an expression for $H(q)H(q^4)$. Then,
	\begin{align*}
	\LHS\eqref{eq:H-modular-eqn}=\frac{\phi(q)\phi(-q^5)-\phi(-q)\phi(q^5)}{2qE(q^2)^2},
	\end{align*}
	thereby confirming \eqref{eq:H-modular-eqn} with recourse to \eqref{eq:modular-deg-5}. Thus, \eqref{eq:triple-AG-2-2} holds true.
\end{proof}

\subsection{$\fsl_4$} We can generalise the analysis above to the case of $\operatorname{PF}_2(\mathfrak{sl}_4)$. In this case, 
extra work is needed to deal with the intertwiners related to the module $L^{2000}_{0020}$ with details given in \cite{Han23,BH24}. In short, we first conjecture that
\begin{align}\label{eq:G4-def}
\mathbf{G}_4:=\frac{1}{2}\begin{pmatrix}
2&1&1&0&1&1\\
1&2&1&1&0&1\\
1&1&2&1&1&2\\
0&1&1&2&1&1\\
1&0&1&1&2&1\\
1&1&2&1&1&2
\end{pmatrix},
\end{align}
and define, for (up to equivalence) one untwisted and two twisted representations,  
\begin{align*}
    \func{\operatorname{UCPF}}{\mathbf{G}_4;\mathbf{0}}&:=q^{-\frac{1}{12}}\sum_{N_i\ge 0}\dfrac{q^{\frac{1}{2}\mathbf{N}^{\mathsf{T}}\cdot\mathbf{G}_4\cdot\mathbf{N}}}{\prod_i (q;q)_{N_i}},\\
    \func{\operatorname{UCPF}}{\mathbf{G}_4;\tfrac{1}{2} (1,0,1,1,0,1)}&:=q^{\frac{1}{6}-\frac{1}{12}}\sum_{N_i\ge 0}\dfrac{q^{\frac{1}{2}\mathbf{N}^{\mathsf{T}}\cdot\mathbf{G}_4\cdot\mathbf{N}-\frac{1}{2}(N_1+N_3+N_4+N_6)}}{\prod_i (q;q)_{N_i}},\\
    \func{\operatorname{UCPF}}{\mathbf{G}_4;\tfrac{1}{2} (1,1,0,0,0,1)}&:=q^{\frac{1}{8}-\frac{1}{12}}\sum_{N_i\ge 0}\dfrac{q^{\frac{1}{2}\mathbf{N}^{\mathsf{T}}\cdot\mathbf{G}_4\cdot\mathbf{N}-\frac{1}{2}(N_1+N_2+N_6)}}{\prod_i (q;q)_{N_i}},
\end{align*}
where
\begin{align}\label{eq:G4-N}
    \mathbf{N}^{\mathsf{T}}=(N_1,N_2,N_3,N_4,N_5,N_6).
\end{align}
Then we obtain more identities by connecting UCPFs and parefermionic characters of $\operatorname{PF}_2(\mathfrak{sl}_4)$. 
\begin{theorem}\label{th:sl4UCPF}
\begin{subequations}
\begin{align}
    \func{\operatorname{UCPF}}{\mathbf{G}_4;\mathbf{0}}&=b^{2000}_{2000}+b^{2000}_{0020}+6b^{2000}_{0101},\label{eq:sl4untwi}\\
    \func{\operatorname{UCPF}}{\mathbf{G}_4;\tfrac{1}{2} (1,0,1,1,0,1)}&=2\left(b^{0101}_{2000}+3b^{0101}_{0101}\right),\label{eq:sl4twi-1}\\
    \func{\operatorname{UCPF}}{\mathbf{G}_4;\tfrac{1}{2} (1,1,0,0,0,1)}&=4\left(b^{1100}_{1100}+b^{1100}_{0011}\right).\label{eq:sl4twi-2}
\end{align}
\end{subequations}

\end{theorem}

\begin{proof}
We start by noticing that \eqref{eq:sl4untwi} is equivalent to
\begin{align}
    \sum_{N_i\ge 0} \frac{q^{\mathbf{N}^{\mathsf{T}}\cdot \mathbf{G}_4\cdot \mathbf{N}}}{\prod_i(q^2;q^2)_{N_i}}=\frac{(q^2;q^2)^3_\infty(q^6;q^6)^5_\infty}{(q;q)^2_\infty(q^3;q^3)^2_\infty(q^4;q^4)^2_\infty(q^{12};q^{12})^2_\infty}+\frac{4q(q^4;q^4)^2_\infty(q^{12};q^{12})^2_\infty}{(q^2;q^2)^3_\infty(q^6;q^6)_\infty}.\label{eq:VOA-1-1}
\end{align}
To see this equivalence, we shall also use the fact that
\begin{align*}
		&\frac{E(q^2)^3E(q^6)^5}{E(q)^2E(q^{3})^2E(q^4)^2E(q^{12})^2}+\frac{4qE(q^4)^2E(q^{12})^2}{E(q^2)^3E(q^6)}\\
		& \quad= \frac{E(q^8)^4E(q^{12})^8}{E(q^2)E(q^{4})^4E(q^6)^3E(q^{24})^4}+\frac{q^2E(q^4)^8E(q^{6})E(q^{24})^4}{E(q^2)^5E(q^{8})^4E(q^{12})^4}+\frac{6qE(q^4)^2E(q^{12})^2}{E(q^2)^3E(q^6)},
	\end{align*}
which follows by substituting \eqref{eq:2-dis-1/E1E3} into the left-hand side of the above. Now, for \eqref{eq:VOA-1-1}, it suffices to use \eqref{eq:VOA-z-1} with $\delta_1=\delta_2=0$ and $z_1=z_2=1$, and then apply \eqref{eq:double-sum-rep} with $q$ replaced by $q^2$.

Next, \eqref{eq:sl4twi-1} is equivalent to
\begin{align}
    \sum_{N_i\ge 0} \frac{q^{\mathbf{N}^{\mathsf{T}}\cdot \mathbf{G}_4\cdot \mathbf{N}-(N_1+N_3+N_4+N_6)}}{\prod_i(q^2;q^2)_{N_i}}=\frac{6(q^3;q^3)^3_\infty}{(q;q)_\infty(q^2;q^2)^2_\infty}.\label{eq:VOA-1-2}
\end{align}
Here the equivalence comes from
\begin{align*}
		\frac{E(q^{3})^3}{E(q)E(q^2)^2} &= \frac{E(q^4)^3E(q^6)^2}{E(q^2)^4E(q^{12})} + \frac{q E(q^{12})^3}{E(q^2)^2 E(q^4)},
	\end{align*}
which is true by invoking \eqref{eq:2-dis-E3^3/E1}. Now, for \eqref{eq:VOA-1-2}, we use \eqref{eq:VOA-z-1} with $\delta_1=1$, $\delta_2=0$ and $(q,z_1,z_2)\mapsto (q^2,q^2,1)$. Then,
	\begin{align*}
		\LHS\eqref{eq:VOA-1-2}&=\frac{1}{(q^2;q^2)_\infty^2}\sum_{M_1,M_2=-\infty}^\infty q^{M_1^2+M_2^2-M_1M_2-M_1}(1+q^{2M_1})\\
		&=\frac{2}{(q^2;q^2)_\infty^2}\sum_{M_1,M_2=-\infty}^\infty q^{M_1^2+M_2^2-M_1M_2-M_1}\\
		\text{\tiny (by \eqref{eq:double-sum-rep})}&=\frac{2}{E(q^2)^2}\left(\frac{E(q^2)^7E(q^{3})E(q^{12})}{E(q)^3E(q^4)^3E(q^6)}+\frac{2E(q^4)^3E(q^6)^2}{E(q^2)^2E(q^{12})}\right)\\
		\text{\tiny (by \eqref{eq:13-1})}&=\frac{6E(q^{3})^3}{E(q)E(q^2)^2}.
	\end{align*}
	
Finally, \eqref{eq:sl4twi-2} is equivalent to
\begin{align}
    \sum_{N_i\ge 0} \frac{q^{\mathbf{N}^{\mathsf{T}}\cdot \mathbf{G}_4\cdot \mathbf{N}-(N_1+N_2+N_6)}}{\prod_i(q^2;q^2)_{N_i}}=\frac{4(q^2;q^2)_\infty^2}{(q;q)_\infty^2}.\label{eq:VOA-2}
\end{align}
For the equivalence, we need
\begin{align*}
	\frac{E(q^2)^2}{E(q)^2} &= \frac{E(q^8)^5}{E(q^2)^3 E(q^{16})^2} + \frac{2qE(q^4)^2 E(q^{16})^2}{E(q^2)^3 E(q^8)},
\end{align*}
which can be shown with recourse to \eqref{eq:2-dis-1/E1^2}. Further, for \eqref{eq:VOA-2}, we only need to take $(z_1,z_2)=(1,1)$ in \eqref{eq:VOA-1-3}.
\end{proof}

\section{Lattice characters for $\mathbb{Z}_2$-orbifolds of scaled $\fsl_n$-root lattices}

\subsection{Fermionic lattice construction}

Denote the root lattice of $\fsl_{n+1}$ by $A_{n}$. There is a standard process to construct a bosonic conformal field theory from $A_{n}$ 
(see for example, Section 15.6.3 in \cite{Fra1997}). We notice that scaling the root lattice $A_n$ by a factor of $1/\sqrt{2}$ 
results in a conformal field theory with coupled free fermions \cite{Han23,BH24}, denoted by $A_n/\sqrt{2}$ . The modules of $A_n/\sqrt{2}$ are in 
correspondence with $\func{}{A_n^*/\sqrt{2}}/\func{}{A_n/\sqrt{2}}$, where $L^*$ denotes the dual lattice of $L$. 
The expressions for characters of these modules are also straightforward. We have
\begin{equation}
    \chi_{\gamma}(\tau)=\eta(\tau)^{-n}\sum_{v\in A_n/\sqrt{2}+\gamma}q^{\frac{1}{2}|v|^2},
\end{equation}
where $\gamma\in \func{}{A_n^*/\sqrt{2}}/\func{}{A_n/\sqrt{2}}$.

For example, in the case of $\fsl_2$, we have
$$\big(A_1^*/\sqrt{2}\big) /  \big(A_1/\sqrt{2}\big) \cong \{ \gamma_0=0,\; \gamma_1= \tfrac{1}{2\sqrt2}\alpha\}  \cong \{ 0, \tfrac12\}, $$
with corresponding characters
\begin{align*}
    \chi_0&=\frac{1}{\eta(\tau)}\sum_{n\in\Bbb{Z}}q^{\frac{1}{2}n^2}=\eta(\tau)^4\eta(\tau/2)^{-2}\eta(2\tau)^{-2},\\
    \chi_1&=\frac{1}{\eta(\tau)}\sum_{n\in\Bbb{Z}}q^{\frac{1}{2}(n+\frac{1}{2})^2}=\eta(\tau)^{-2}\eta(2\tau)^{2}. 
\end{align*}

In general, it turns out that in order to recover the full set of string functions, we need to consider a $\mathbb{Z}_2$ orbifold of the above lattice construction, involving both twisted sectors as well as twist projections. This is described in detail in \cite{Han23,BH24}. In the following, we simply list several resulting characters and their relations to string functions and UCPFs.

\subsection{$\fsl_3$} 
Let $\{\alpha_1,\alpha_2\}$ be simple roots of $\fsl_3$ and they form a basis of $A_2/\sqrt{2}$. 
Then the dual lattice $A_2^*/\sqrt{2}$ has a basis $\big\{\frac{1}{3\sqrt2}(\alpha_1-\alpha_2),\frac{1}{\sqrt{2}}\alpha_1\big\}$, 
so there are three cosets in $\func{}{A_2^*/\sqrt{2}}/\func{}{A_2/\sqrt{2}}$, represented by $\gamma_0:=0$, 
$\gamma_1:=\frac{1}{3\sqrt{2}}(\alpha_1-\alpha_2)$ and $\gamma_2:=\frac{2}{3\sqrt{2}}(\alpha_1-\alpha_2)$. 

The characters of the orbifold modules without the twist factor are
\begin{align*}
    \chi_{\gamma_0}(\tau)&=\eta(\tau)^{-2}\sum_{m,n\in\mathbb{Z}}q^{\frac{1}{2}(m^2+n^2-mn)},\\
     \chi_{\gamma_1}(\tau)= \chi_{\gamma_2}(\tau)&=\eta(\tau)^{-2}q^{\frac{2}{3}}\sum_{m,n\in\mathbb{Z}}q^{\frac{1}{2}(m^2+n^2-mn+2m)},\\
    \chi_{\operatorname{orb}}(\tau)&=\eta(\tau)^2\eta(\tau/2)^{-2}.
\end{align*}
We notice that they are closely related to the UCPFs of $\operatorname{PF}_2(\mathfrak{sl}_4)$, as we should expect by studying the conformal field structure of these two theories. 
\begin{theorem}\label{th:A2lat}
\begin{subequations}
\begin{align}
    \chi_{\gamma_0}&=\func{\operatorname{UCPF}}{\mathbf{G}_4;\mathbf{0}},\label{eq:A3unshf}\\
    2\chi_{\gamma_1}&=\func{\operatorname{UCPF}}{\mathbf{G}_4;\tfrac{1}{2} (1,0,1,1,0,1)},\label{eq:A3shf}\\
    4\chi_{\operatorname{orb}}&=\func{\operatorname{UCPF}}{\mathbf{G}_4;\tfrac{1}{2} (1,1,0,0,0,1)}.\label{eq:A3orb}
\end{align}
\end{subequations}

\end{theorem}

\begin{proof}
We obtain \eqref{eq:A3unshf} by taking $\delta_1=\delta_2=0$ and $z_1=z_2=1$ in \eqref{eq:VOA-z-1}. For \eqref{eq:A3shf}, we know from \eqref{eq:VOA-z-1} that
	\begin{align*}
		\func{\operatorname{UCPF}}{\mathbf{G}_4;\tfrac{1}{2}(1,0,1,1,0,1)}&=\frac{2q^{\frac{1}{6}}}{\eta(\tau)^2}\sum_{m,n=-\infty}^\infty q^{\frac{m^2}{2}+\frac{n^2}{2}-\frac{mn}{2}-\frac{m}{2}}.
\end{align*}
Making the change of variables $(m,n)\mapsto (n-m,-m-1)$ in the above gives the desired result. Finally, \eqref{eq:A3orb} is a direct consequence of \eqref{eq:VOA-2}.
\end{proof}

The characters of the orbifold modules, with the twist factor inserted, are
\begin{subequations}
\begin{align}
    \chi_{\gamma_0}^T&=\eta^{-2}\sum_{m,n\in\mathbb{Z}}(-1)^{m+n}q^{\frac{1}{2}(m^2+n^2-mn)},\label{eq:chi-gamma-0-T}\\
    \chi_{\gamma_1}^T= \chi_{\gamma_2}^T&= \eta^{-2}q^{\frac{2}{3}}\sum_{m,n\in\mathbb{Z}}(-1)^{m+n}q^{\frac{1}{2}(m^2+n^2-mn+2m)},\label{eq:chi-gamma-1-T}\\
    \chi_{\operatorname{orb}}^T&=\eta(2\tau)^2\eta(\tau/2)^2\eta(\tau)^{-4}.
\end{align}
\end{subequations}

We again derive some identities between these lattice characters and the parafermionic characters. 
\begin{theorem}
\begin{subequations}
\begin{align}
    \chi_{\gamma_0}^T&=b^{2000}_{2000}+b^{2000}_{0020}-b^{2000}_{0101},\label{eq:A3Tunshf}\\
    \chi_{\gamma_1}^T&=b^{0101}_{2000}-b^{0101}_{0101},\label{eq:A3Tshf}\\
     \chi_{\operatorname{orb}}^T&=b^{1100}_{1100}-b^{1100}_{0011}.\label{eq:A3Torb}
\end{align}
\end{subequations}
\end{theorem}

\begin{proof}
We see that the right-hand side of \eqref{eq:chi-gamma-0-T} is an instance of \eqref{eq:double-sum-rep}:
\begin{align*}
    \sum_{m,n\in\mathbb{Z}}(-1)^{m+n}q^{\frac{1}{2}(m^2+n^2-mn)}=\frac{E(q^{\frac{1}{2}})^2E(q^{\frac{3}{2}})^2}{E(q)E(q^3)}.
\end{align*}
Therefore, \eqref{eq:A3Tunshf} is equivalent to
\begin{align*}
		\frac{E(q)^2E(q^3)^2}{E(q^2)^3E(q^6)}=\frac{E(q^8)^4E(q^{12})^8}{E(q^2)E(q^{4})^4E(q^6)^3E(q^{24})^4}+
		\frac{q^2E(q^4)^8E(q^{6})E(q^{24})^4}{E(q^2)^5E(q^{8})^4E(q^{12})^4}-\frac{2qE(q^4)^2E(q^{12})^2}{E(q^2)^3E(q^6)},
\end{align*}
which follows by substituting \eqref{eq:2-dis-E1E3} into the left-hand side of the above.

For the right-hand side of \eqref{eq:chi-gamma-1-T}, we first change the variables $(m,n)\mapsto (-n-1,m-n-1)$. Thus,
\begin{align*}
    \sum_{m,n\in\mathbb{Z}}(-1)^{m+n}q^{\frac{1}{2}(m^2+n^2-mn+2m)} &= q^{-\frac{1}{2}}\sum_{m,n\in\mathbb{Z}}(-1)^{m}q^{\frac{1}{2}(m^2+n^2-mn-m)}\\
    \text{\tiny (by \eqref{eq:double-sum-rep})}&=q^{-\frac{1}{2}}\left(\frac{E(q)^7E(q^{\frac{3}{2}})E(q^6)}{E(q^{\frac{1}{2}})^3E(q^2)^3E(q^3)}-\frac{2E(q^2)^3E(q^3)^2}{E(q)^2E(q^6)}\right)\\
	\text{\tiny (by \eqref{eq:13-2})}&=-\frac{q^{-\frac{1}{2}}E(q^{\frac{1}{2}})^3 E(q^3)^2}{E(q)^2 E(q^{\frac{3}{2}})}.
\end{align*}
We then see that \eqref{eq:A3Tshf} is equivalent to
\begin{align*}
		\frac{E(q)^3 E(q^6)^2}{E(q^2)^4 E(q^{3})} = \frac{E(q^4)^3E(q^6)^2}{E(q^2)^4E(q^{12})} - \frac{3q E(q^{12})^3}{E(q^2)^2 E(q^4)},
\end{align*}
which is true by invoking \eqref{eq:2-dis-E1^3/E3}.

Finally, \eqref{eq:A3Torb} is equivalent to
\begin{align*}
		\frac{E(q)^2E(q^4)^2}{E(q^2)^4} = \frac{E(q^8)^5}{E(q^2)^3 E(q^{16})^2} - \frac{2qE(q^4)^2 E(q^{16})^2}{E(q^2)^3 E(q^8)},
\end{align*}
which can be shown with recourse to \eqref{eq:2-dis-E1^2}.
\end{proof}

\section{A $q$-rious perspective}\label{sec:transform}

As we will see in this section, the previous results can be nicely established from a $q$-series perspective. They are also of great interest within the community of $q$-theorists for their deep connections with identities of Rogers--Ramanujan type.

We begin with a general transform that rewrites a triple summation related to $\mathbf{G}_3$ in \eqref{eq:G3-def} into a single summation.

\begin{theorem}\label{th:G3-general}
Let $\mathbf{G}_3$ and $\mathbf{N}$ be as in \eqref{eq:G3-def} and \eqref{eq:G3-N}, respectively. Then,
\begin{align}\label{eq:G3-general}
&\sum_{N_1,N_2,N_3\ge 0} \frac{z_1^{N_1}z_2^{N_2}z_3^{N_3}q^{\mathbf{N}^{\mathsf{T}}\cdot\mathbf{G}_3\cdot\mathbf{N}}}{(q^2;q^2)_{N_1}(q^2;q^2)_{N_2}(q^2;q^2)_{N_3}}\notag\\
&\qquad= (-z_1q;q^2)_\infty \sum_{M\ge 0}\frac{z_2^M z_3^M q^{3M^2} (-z_2 z_3^{-1}q;q^2)_M (-z_2^{-1}z_3q;q^2)_M}{(-z_1q;q^2)_M (q^2;q^2)_{2M}}\notag\\
&\qquad\quad+ (-z_1 q^2;q^2)_\infty \sum_{M\ge 0} \frac{(z_2+z_3) z_2^M z_3^M q^{3M^2+3M+1} (-z_2z_3^{-1}q^2;q^2)_M (-z_2^{-1}z_3q^2;q^2)_M}{(-z_1q^2;q^2)_{M} (q^2;q^2)_{2M+1}}.
\end{align}
\end{theorem}

For $\mathbf{G}_4$, we are facing sextuple summations, but we unfortunately have to give up some generality. 
That is, to bring out valid transforms, it is necessary to reduce the number of free parameters. However, 
we still get compensated by two sets of entirely different transforms to be established. In our first result, the sextuple summation is rewritten as a double bilateral sum.

\begin{theorem}\label{th:VOA-z}
	Let $\mathbf{G}_4$ and $\mathbf{N}$ be as in \eqref{eq:G4-def} and \eqref{eq:G4-N}, respectively. Then for $\delta_1,\delta_2\in\{0,1\}$,
	\begin{align}
		&\sum_{N_1,N_2,N_3,N_4,N_5,N_6\ge 0} \frac{z_1^{N_1+N_3-N_4-N_6}z_2^{N_2+N_3-N_5-N_6}
		q^{\mathbf{N}^{\mathsf{T}}\cdot \mathbf{G}_4\cdot \mathbf{N}-2\delta_1(N_4+N_6)-2\delta_2(N_5+N_6)} }{(q^2;q^2)_{N_1}
		(q^2;q^2)_{N_2}(q^2;q^2)_{N_3}(q^2;q^2)_{N_4}(q^2;q^2)_{N_5}(q^2;q^2)_{N_6}}\notag\\
		&\qquad = \frac{1}{(q^2;q^2)_\infty^2}\sum_{M_1,M_2=-\infty}^\infty z_1^{M_1} z_2^{M_2}q^{M_1^2+M_2^2-M_1M_2}
		(1+\delta_1 q^{2M_1}) (1+\delta_2 q^{2M_2}).\label{eq:VOA-z-1}
	\end{align}
\end{theorem}

For the second result, its flavour is similar to Theorem \ref{th:G3-general}.

\begin{theorem}\label{th:VOA-1}
	Let $\mathbf{G}_4$ and $\mathbf{N}$ be as in \eqref{eq:G4-def} and \eqref{eq:G4-N}, respectively. Then,
	\begin{subequations}
	\begin{align}
		&\sum_{N_1,N_2,N_3,N_4,N_5,N_6\ge 0} \frac{z_1^{N_1+N_3+N_4+N_6} z_2^{N_2+N_3+N_5+N_6} q^{\mathbf{N}^{\mathsf{T}}\cdot \mathbf{G}_4\cdot \mathbf{N}-(N_1+N_2+N_6)}}{(q^2;q^2)_{N_1}(q^2;q^2)_{N_2}(q^2;q^2)_{N_3}(q^2;q^2)_{N_4}(q^2;q^2)_{N_5}(q^2;q^2)_{N_6}}\notag\\
		&\qquad=\frac{(-z_1z_2;q)_\infty}{(q;q^2)_\infty}\sum_{M\ge 0}\frac{(-1)^M q^{M^2} (z_1^2;q^2)_M(z_2^2;q^2)_M}{(q^2;q^2)_M (-z_1 z_2;q)_{2M}}\notag\\
		&\qquad\quad+\frac{(-z_1z_2;q)_\infty}{(q;q^2)_\infty}\sum_{M\ge 0}\frac{(-1)^M \big(z_1+z_2\big)q^{M^2+2M} (z_1^2;q^2)_M(z_2^2;q^2)_M}{(q^2;q^2)_M (-z_1 z_2;q)_{2M+1}},\label{eq:VOA-1-3}\\
		&\sum_{N_1,N_2,N_3,N_4,N_5,N_6\ge 0} \frac{(-1)^{N_1+N_2+N_6}z_1^{N_1+N_3+N_4+N_6} z_2^{N_2+N_3+N_5+N_6} q^{\mathbf{N}^{\mathsf{T}}\cdot \mathbf{G}_4\cdot \mathbf{N}-(N_1+N_2+N_6)}}{(q^2;q^2)_{N_1}(q^2;q^2)_{N_2}(q^2;q^2)_{N_3}(q^2;q^2)_{N_4}(q^2;q^2)_{N_5}(q^2;q^2)_{N_6}}\notag\\
		&\qquad=\frac{(z_1z_2;-q)_\infty}{(-q;q^2)_\infty}\sum_{M\ge 0}\frac{q^{M^2} (z_1^2;q^2)_M(z_2^2;q^2)_M}{(q^2;q^2)_M (z_1 z_2;-q)_{2M}}\notag\\
		&\qquad\quad-\frac{(-z_1z_2;-q)_\infty}{(-q;q^2)_\infty}\sum_{M\ge 0}\frac{\big(z_1+z_2\big) q^{M^2+2M} (z_1^2;q^2)_M(z_2^2;q^2)_M}{(q^2;q^2)_M (-z_1 z_2;-q)_{2M+1}}.\label{eq:VOA-1-3-T-1}
	\end{align}
	\end{subequations}
\end{theorem}

\subsection{Proof of Theorem \ref{th:G3-general}}

Starting with the $N_1$-summation, we have
\begin{align*}
    \LHS\eqref{eq:G3-general}&=\sum_{N_2,N_3\ge 0} \frac{z_2^{N_2}z_3^{N_3}q^{\binom{N_2+N_3}{2}+\binom{N_2}{2}+\binom{N_3}{2}+N_2+N_3}}{(q^2;q^2)_{N_2}(q^2;q^2)_{N_3}}\sum_{N_1\ge 0}\frac{z_1^{N_1}q^{2\binom{N_1}{2}+(N_2+N_3+1)N_1}}{(q^2;q^2)_{N_1}}\\
	\text{\tiny (by \eqref{eq:Eul-2})}&= \sum_{N_2,N_3\ge 0} \frac{z_2^{N_2}z_3^{N_3}q^{\binom{N_2+N_3}{2}+\binom{N_2}{2}+\binom{N_3}{2}+N_2+N_3}}{(q^2;q^2)_{N_2}(q^2;q^2)_{N_3}}(-z_1q^{N_2+N_3+1};q^2)_\infty\\
	\text{\tiny ($M=N_2+N_3$)}&= \sum_{M\ge 0}(-z_1q^{M+1};q^2)_\infty q^{\binom{M}{2}+M}\sum_{\substack{N_2,N_3\ge 0\\N_2+N_3=M}}\frac{z_2^{N_2}z_3^{N_3}q^{\binom{N_2}{2}+\binom{N_3}{2}}}{(q^2;q^2)_{N_2}(q^2;q^2)_{N_3}}\\
	\text{\tiny (by \eqref{eq:SM-exp})}&=\sum_{M\ge 0}(-z_1q^{2M+1};q^2)_\infty q^{\binom{2M}{2}+2M}\\
	&\quad\quad\times\frac{z_2^M z_3^M q^{M^2-M} (-z_2 z_3^{-1}q;q^2)_M (-z_2^{-1}z_3q;q^2)_M}{(q^2;q^2)_{2M}}\\
	&\quad + \sum_{M\ge 0}(-z_1q^{2M+2};q^2)_\infty q^{\binom{2M+1}{2}+(2M+1)}\\
	&\quad\quad\times \frac{(z_2+z_3) z_2^M z_3^M q^{M^2} (-z_2z_3^{-1}q^2;q^2)_M (-z_2^{-1}z_3q^2;q^2)_M}{(q^2;q^2)_{2M+1}}\\
	&=(-z_1q;q^2)_\infty \sum_{M\ge 0}\frac{z_2^M z_3^M q^{3M^2} (-z_2 z_3^{-1}q;q^2)_M (-z_2^{-1}z_3q;q^2)_M}{(-z_1q;q^2)_M (q^2;q^2)_{2M}}\\
	&\quad+ (-z_1 q^2;q^2)_\infty \sum_{M\ge 0} \frac{(z_2+z_3) z_2^M z_3^M q^{3M^2+3M+1} (-z_2z_3^{-1}q^2;q^2)_M (-z_2^{-1}z_3q^2;q^2)_M}{(-z_1q^2;q^2)_{M} (q^2;q^2)_{2M+1}},
\end{align*}
thereby establishing the desired result. \qed

\subsection{Proof of Theorem \ref{th:VOA-z}}

We begin with the following change of variables
\begin{align}\label{eq:N-M-sub}
	\left\{
	\begin{aligned}
	M_1&=N_1+N_3-N_4-N_6,\\
	M_2&=N_2+N_3-N_5-N_6,\\
	M_3&=N_4+N_6,\\
	M_4&=N_5+N_6,
	\end{aligned}
	\right.
	\quad\Longleftrightarrow\quad
	\left\{
	\begin{aligned}
	N_1&=M_1+M_3-N_3,\\
	N_2&=M_2+M_4-N_3,\\
	N_4&=M_3-N_6,\\
	N_5&=M_4-N_6.
	\end{aligned}
	\right.
\end{align}
It follows that the summation on the left-hand side of \eqref{eq:VOA-z-1} can be rewritten over the indices $(M_1,M_2,M_3,M_4,N_3,N_6)$:
\begin{align*}
	\LHS\eqref{eq:VOA-z-1}&=\sum_{M_1,M_2=-\infty}^\infty z_1^{M_1} z_2^{M_2} \sum_{M_3,M_4\ge 0} q^{-2\delta_1 M_3-2\delta_2 M_4}\\
	&\quad\times q^{M_1^2+M_2^2+2M_3^2+2M_4^2+M_1M_2+2M_1M_3+2M_1M_4+2M_2M_3+2M_2M_4+4M_3M_4}\\
	&\quad\times\sum_{N_3\ge 0}\frac{q^{2N_3^2-2N_3(M_1+M_2+M_3+M_4)}}{(q^2;q^2)_{N_3}(q^2;q^2)_{M_1+M_3-N_3}(q^2;q^2)_{M_2+M_4-N_3}}\\
	&\quad\times\sum_{N_6\ge 0}\frac{q^{2N_6^2-2N_6(M_3+M_4)}}{(q^2;q^2)_{N_6}(q^2;q^2)_{M_3-N_6}(q^2;q^2)_{M_4-N_6}}\\
	\text{\tiny (by \eqref{eq:q-CV-spe})}&=\sum_{M_1,M_2=-\infty}^\infty z_1^{M_1} z_2^{M_2} \sum_{M_3,M_4\ge 0} q^{-2\delta_1 M_3-2\delta_2 M_4}\\
	&\quad\times q^{M_1^2+M_2^2+2M_3^2+2M_4^2+M_1M_2+2M_1M_3+2M_1M_4+2M_2M_3+2M_2M_4+4M_3M_4}\\
	&\quad\times\frac{q^{-2(M_1+M_3)(M_2+M_4)}}{(q^2;q^2)_{M_1+M_3}(q^2;q^2)_{M_2+M_4}}\cdot \frac{q^{-2M_3M_4}}{(q^2;q^2)_{M_3}(q^2;q^2)_{M_4}}\\
	&=\sum_{M_1,M_2=-\infty}^\infty z_1^{M_1} z_2^{M_2} q^{M_1^2+M_2^2-M_1M_2}\\
	&\quad\times \sum_{M_3\ge 0}\frac{q^{2M_3^2+2(M_1-\delta_1)M_3}}{(q^2;q^2)_{M_3}(q^2;q^2)_{M_1+M_3}}\sum_{M_4\ge 0}\frac{q^{2M_4^2+2(M_2-\delta_2)M_4}}{(q^2;q^2)_{M_4}(q^2;q^2)_{M_2+M_4}}\\
	\text{\tiny (by \eqref{eq:q-Gauss-spe})}&=\frac{1}{(q^2;q^2)_\infty^2}\sum_{M_1,M_2=-\infty}^\infty z_1^{M_1} z_2^{M_2}q^{M_1^2+M_2^2-M_1M_2}(1+\delta_1 q^{2M_1}) (1+\delta_2 q^{2M_2}),
\end{align*}
confirming the desired result. \qed

\subsection{Proof of Theorem \ref{th:VOA-1}}

The proof of Theorem \ref{th:VOA-1} is entirely different from that for Theorem \ref{th:VOA-z}. For \eqref{eq:VOA-1-3}, starting with the inner summation in $N_6$ by using \eqref{eq:Eul-2}, we have, by also recalling Lemma \ref{le:T-Exp},
	\begin{align*}
	&\LHS\eqref{eq:VOA-1-3}\\
	&=\sum_{N_1,N_2,N_3,N_4,N_5\ge 0}\frac{z_1^{N_1+N_3+N_4} z_2^{N_2+N_3+N_5}}{(q^2;q^2)_{N_1}(q^2;q^2)_{N_2}(q^2;q^2)_{N_3}(q^2;q^2)_{N_4}(q^2;q^2)_{N_5}}\notag\\
	&\quad\times q^{N_1^2+N_2^2+N_3^2+N_4^2+N_5^2+N_1N_2+N_1N_3+N_1N_5+N_2N_3+N_2N_4+N_3N_4+N_3N_5+N_4N_5-N_1-N_2}\\
	&\quad \times\sum_{N_6\ge 0}\frac{(z_1z_2)^{N_6}q^{2\binom{N_6}{2}+(N_1+N_2+2N_3+N_4+N_5)N_6}}{(q^2;q^2)_{N_6}}\\
	&=\sum_{N_1,N_2,N_3,N_4,N_5\ge 0}\frac{z_1^{N_1+N_3+N_4} z_2^{N_2+N_3+N_5}}{(q^2;q^2)_{N_1}(q^2;q^2)_{N_2}(q^2;q^2)_{N_3}(q^2;q^2)_{N_4}(q^2;q^2)_{N_5}}\notag\\
	&\quad\times q^{N_1^2+N_2^2+N_3^2+N_4^2+N_5^2+N_1N_2+N_1N_3+N_1N_5+N_2N_3+N_2N_4+N_3N_4+N_3N_5+N_4N_5-N_1-N_2}\\
	&\quad \times (-z_1z_2q^{N_1+N_2+2N_3+N_4+N_5};q^2)_\infty\\
	&=\sum_{M\ge 0}(-z_1z_2q^{M};q^2)_\infty T_M\\
	&=\sum_{M\ge 0}(-z_1z_2q^{2M};q^2)_\infty T_{2M}+\sum_{M\ge 0}(-z_1z_2q^{2M+1};q^2)_\infty T_{2M+1}\\
	&=(-z_1z_2;q^2)_\infty \sum_{M\ge 0} \frac{z_1^M z_2^Mq^{M^2}(-z_1 z_2^{-1};q^2)_M (-z_1^{-1} z_2;q^2)_M}{(q^2;q^2)_{M}(q;q^2)_{M}(-z_1z_2;q^2)_M}\\
	&\quad+(-z_1z_2q;q^2)_\infty \sum_{M\ge 0} \frac{\big(z_1+z_2\big)z_1^M z_2^Mq^{M(M+1)}(-z_1 z_2^{-1} q;q^2)_M (-z_1^{-1} z_2 q;q^2)_M}{(q^2;q^2)_{M}(q;q^2)_{M+1}(-z_1z_2q;q^2)_M}\\
	&=(-z_1z_2;q^2)_\infty\cdot \lim_{\tau\to 0} {}_{3}\phi_2\left(\begin{matrix} 1/\tau,-z_1 z_2^{-1}, -z_1^{-1} z_2\\ -z_1 z_2,q\end{matrix}; q^2, -z_1z_2q\tau\right)\\
	&\quad+(-z_1z_2q;q^2)_\infty\cdot \frac{z_1+z_2}{1-q}\cdot \lim_{\tau\to 0} {}_{3}\phi_2\left(\begin{matrix} 1/\tau,-z_1 z_2^{-1}q, -z_1^{-1} z_2q\\ -z_1 z_2q,q^3\end{matrix}; q^2, -z_1z_2q^2\tau\right)\\
	&=(-z_1z_2;q^2)_\infty\cdot \frac{(-z_1z_2q;q^2)_\infty}{(q;q^2)_\infty} \lim_{\tau\to 0} {}_{3}\phi_2\left(\begin{matrix} 1/\tau,z_2^2, z_1^2\\ -z_1 z_2,-z_1z_2q\end{matrix}; q^2, q\tau\right)\\
	&\quad+(-z_1z_2q;q^2)_\infty\cdot \frac{z_1+z_2}{1-q}\cdot \frac{(-z_1z_2q^2;q^2)_\infty}{(q^3;q^2)_\infty}\lim_{\tau\to 0} {}_{3}\phi_2\left(\begin{matrix} 1/\tau,z_2^2, z_1^2\\ -z_1 z_2q,-z_1z_2q^2\end{matrix}; q^2, q^3\tau\right)\\
	&=\frac{(-z_1z_2;q)_\infty}{(q;q^2)_\infty}\sum_{M\ge 0}\frac{(-1)^M q^{M^2} (z_1^2;q^2)_M(z_2^2;q^2)_M}{(q^2;q^2)_M (-z_1 z_2;q)_{2M}}\\
	&\quad+\frac{(-z_1z_2;q)_\infty}{(q;q^2)_\infty}\sum_{M\ge 0}\frac{(-1)^M \big(z_1+z_2\big)q^{M^2+2M} (z_1^2;q^2)_M(z_2^2;q^2)_M}{(q^2;q^2)_M (-z_1 z_2;q)_{2M+1}},
	\end{align*}
	confirming \eqref{eq:VOA-1-3}, where we have made use of \eqref{eq:3phi2-1}.
	
	For \eqref{eq:VOA-1-3-T-1}, we also calculate the inner summation over $N_6$ by using \eqref{eq:Eul-2} and then apply Lemma \ref{le:T-Exp}. Thus,
	\begin{align*}
		&\LHS\eqref{eq:VOA-1-3-T-1}\\
		&=\sum_{N_1,N_2,N_3,N_4,N_5\ge 0}\frac{(-1)^{N_1+N_2}z_1^{N_1+N_3+N_4} z_2^{N_2+N_3+N_5}}{(q^2;q^2)_{N_1}(q^2;q^2)_{N_2}(q^2;q^2)_{N_3}(q^2;q^2)_{N_4}(q^2;q^2)_{N_5}}\notag\\
		&\quad\times q^{N_1^2+N_2^2+N_3^2+N_4^2+N_5^2+N_1N_2+N_1N_3+N_1N_5+N_2N_3+N_2N_4+N_3N_4+N_3N_5+N_4N_5-N_1-N_2}\\
		&\quad \times (z_1z_2q^{N_1+N_2+2N_3+N_4+N_5};q^2)_\infty\\
		&=\sum_{M\ge 0}(z_1z_2q^{M};q^2)_\infty \widetilde{T}_M\\
		&=\sum_{M\ge 0}(z_1z_2q^{2M};q^2)_\infty \widetilde{T}_{2M}+\sum_{M\ge 0}(z_1z_2q^{2M+1};q^2)_\infty \widetilde{T}_{2M+1}\\
		&=(z_1z_2;q^2)_\infty \sum_{M\ge 0} \frac{z_1^M z_2^Mq^{M^2}(z_1 z_2^{-1};q^2)_M (z_1^{-1} z_2;q^2)_M}{(q^2;q^2)_{M}(-q;q^2)_{M}(z_1z_2;q^2)_M}\\
		&\quad-(z_1z_2q;q^2)_\infty \sum_{M\ge 0} \frac{\big(z_1+z_2\big)z_1^M z_2^Mq^{M(M+1)}(z_1 z_2^{-1} q;q^2)_M (z_1^{-1} z_2 q;q^2)_M}{(q^2;q^2)_{M}(-q;q^2)_{M+1}(z_1z_2q;q^2)_M}\\
		&=(z_1z_2;q^2)_\infty\cdot \lim_{\tau\to 0} {}_{3}\phi_2\left(\begin{matrix} 1/\tau,z_1 z_2^{-1}, z_1^{-1} z_2\\ z_1 z_2,-q\end{matrix}; q^2, -z_1z_2q\tau\right)\\
		&\quad-(z_1z_2q;q^2)_\infty\cdot \frac{z_1+z_2}{1+q}\cdot \lim_{\tau\to 0} {}_{3}\phi_2\left(\begin{matrix} 1/\tau,z_1 z_2^{-1}q, z_1^{-1} z_2q\\ z_1 z_2q,-q^3\end{matrix}; q^2, -z_1z_2q^2\tau\right)\\
		&=(z_1z_2;q^2)_\infty\cdot \frac{(-z_1z_2q;q^2)_\infty}{(-q;q^2)_\infty} \lim_{\tau\to 0} {}_{3}\phi_2\left(\begin{matrix} 1/\tau,z_2^2, z_1^2\\ z_1 z_2,-z_1z_2q\end{matrix}; q^2, -q\tau\right)\\
		&\quad-(z_1z_2q;q^2)_\infty\cdot \frac{z_1+z_2}{1+q}\cdot \frac{(-z_1z_2q^2;q^2)_\infty}{(-q^3;q^2)_\infty}\lim_{\tau\to 0} {}_{3}\phi_2\left(\begin{matrix} 1/\tau,z_2^2, z_1^2\\ z_1 z_2q,-z_1z_2q^2\end{matrix}; q^2, -q^3\tau\right)\\
		&=\frac{(z_1z_2;-q)_\infty}{(-q;q^2)_\infty}\sum_{M\ge 0}\frac{q^{M^2} (z_1^2;q^2)_M(z_2^2;q^2)_M}{(q^2;q^2)_M (z_1 z_2;-q)_{2M}}\\
		&\quad-\frac{(-z_1z_2;-q)_\infty}{(-q;q^2)_\infty}\sum_{M\ge 0}\frac{\big(z_1+z_2\big) q^{M^2+2M} (z_1^2;q^2)_M(z_2^2;q^2)_M}{(q^2;q^2)_M (-z_1 z_2;-q)_{2M+1}},
	\end{align*}
	confirming \eqref{eq:VOA-1-3-T-1}, where we have also applied \eqref{eq:3phi2-1}. \qed

\section{Outlook}

As we have seen above, the universal chiral partition function can be a bridge to obtain closed forms of string functions, 
which is generally a hard question. We are interested in generalising this method to higher rank affine Lie algebras, 
but currently we are facing some obstacles in determining the statistical interaction matrix $\mathbf{G}$ for higher 
rank cases analogously. We may need more mathematical tools to achieve the goal. 

It is also intriguing that we have found an equivalence between the coset construction $\operatorname{PF}_2(\mathfrak{sl}_4)$ and the $\Bbb{Z}_2$-orbifold 
of the lattice construction $A_2/\sqrt{2}$, which integrates three major methods of constructing conformal field theories. 
The structure of these two theories is examined in more detail in \cite{Han23,BH24}.

From the $q$-series side, this work provides some insights into the modular properties of $q$-summations. It is in general not easy to 
construct $q$-summations with modular properties, but with the considerations of Kac--Moody algebra, one may get a basic idea of 
how these summations would look like, thereby even leading to interesting $q$-hypergeometric transforms.

Another dynamic area that has a close bond with the Rogers--Ramanujan type identities is the theory of integer partitions. 
It is known that one of the two original identities due to Rogers \cite{Rog1894b} and Ramanujan \cite{Ram1919},
\begin{align*}
    \sum_{n\ge 0}\frac{q^{n^2}}{(q;q)_n} = \frac{1}{(q,q^4;q^5)_\infty},
\end{align*}
has the following partition-theoretic interpretation (see Corollary 7.6 in \cite{And1976}):
\begin{quote}
    \textit{The number of partitions of a nonnegative integer $n$ into parts congruent to $\pm 1$ modulo $5$ is the same as the number of 
    partitions of $n$ such that the adjacent parts differ by at least $2$.}
\end{quote}
We also expect that the summation side of our identities of Rogers--Ramanujan type would have a similar meaning as partitions, 
perhaps under analogous but more complicated gap conditions for the parts.

Finally, some conjectural Rogers--Ramanujan type identities in terms of integer partitions were proposed in a recent seminal 
paper due to Kanade and Russell \cite{KR2015}. They had pointed out that three of their Mod $9$ Conjectures, which remain open, 
are related to the principally specialised characters of level 3 standard modules for the affine Lie algebra $D_4^{(3)}$. More recently, 
the analytic counterparts of these conjectures were independently constructed by Kur\c{s}ung\"{o}z \cite{Kur2019}, and Chern and Li \cite{CL2020}. 
Again, these analytic identities have the form of ``$\text{$q$-summation} = \text{$q$-product}$.'' We hope that a deeper look 
at the summations in the analytic counterparts of the Kanade--Russell Conjectures from the Lie-algebraic side would reveal more information and therefore open a window for their proofs.

\appendix
\section{$q$-Series prerequisites}

For the proofs of the transforms in Section \ref{sec:transform}, as well as the aforementioned Rogers--Ramanujan type identities, 
it is necessary to record some results from the standard theory of $q$-series. Recall that the \emph{$q$-hypergeometric function} ${}_{r}\phi_s$ is defined by
\begin{align*}
{}_{r}\phi_s\left(\begin{matrix} A_1,A_2,\ldots,A_r\\ B_1,B_2,\ldots,B_s \end{matrix}; q, z\right):=\sum_{n\ge 0} \frac{(A_1,A_2,\ldots,A_r;q)_n \big((-1)^n 
q^{\binom{n}{2}}\big)^{s-r+1} z^n}{(q,B_1,B_2,\ldots,B_{s};q)_n}.
\end{align*}
Meanwhile, we introduce \emph{Ramanujan's theta functions}:
\begin{subequations}
\begin{align}
\phi(q)&:=\sum_{n=-\infty}^\infty q^{n^2}=\frac{(q^2;q^2)^5_\infty}{(q;q)^2_\infty(q^4;q^4)^2_\infty},\label{eq:phi-def}\\
\psi(q)&:=\sum_{n\ge 0} q^{\frac{n(n+1)}{2}}=\frac{(q^2;q^2)^2_\infty}{(q;q)_\infty},\label{eq:psi-def}
\end{align}
\end{subequations}
and the \emph{Rogers--Ramanujan functions}:
\begin{subequations}
\begin{align}
G(q)&:=\frac{1}{(q,q^4;q^5)_\infty},\label{eq:G-def}\\
H(q)&:=\frac{1}{(q^2,q^3;q^5)_\infty}.\label{eq:H-def}
\end{align}
\end{subequations}

\begin{lemma}[Euler's summations]
	We have
	\begin{subequations}
	\begin{align}
	\frac{1}{(z;q)_\infty} &= \sum_{n\ge 0}\frac{z^n}{(q;q)_n},\label{eq:Eul-1}\\
	(z;q)_\infty &= \sum_{n\ge 0}\frac{(-1)^n z^n q^{\binom{n}{2}}}{(q;q)_n}.\label{eq:Eul-2}
	\end{align}
	\end{subequations}
\end{lemma}

\begin{proof}
	See Andrews \cite[(2.2.5) and (2.2.6)]{And1976}.
\end{proof}

\begin{lemma}[Jacobi's triple product identity]
	We have
	\begin{align}\label{eq:JTP}
	(q,z,q/z;q)_\infty = \sum_{n=-\infty}^\infty (-1)^n z^n q^{\binom{n}{2}}.
	\end{align}
\end{lemma}

\begin{proof}
	See Andrews \cite[(2.2.10)]{And1976}.
\end{proof}

\begin{lemma}[$q$-Binomial theorem]
	We have
	\begin{align}\label{eq:q-Bin}
	(z;q)_N = \sum_{n\ge 0} (-1)^n z^n q^{\binom{n}{2}}{N \brack n}_q,
	\end{align}
	where
	\begin{align*}
		{M \brack L}_q:=\begin{cases}
			\dfrac{(q;q)_M}{(q;q)_L(q;q)_{M-L}}, & \text{if $0\le L\le M$},\\[6pt]
			0, & \text{otherwise}.
		\end{cases}
	\end{align*}
\end{lemma}

\begin{proof}
	See Andrews \cite[(3.3.6)]{And1976}.
\end{proof}

\begin{lemma}
	We have
	\begin{subequations}
	\begin{align}
	\frac{1}{E(q)^2}&=\frac{E(q^8)^5}{E(q^2)^5 E(q^{16})^2}+\frac{2qE(q^4)^2E(q^{16})^2}{E(q^2)^5E(q^8)},\label{eq:2-dis-1/E1^2}\\
	E(q)^2&=\frac{E(q^2)E(q^8)^5}{E(q^4)^2 E(q^{16})^2}-\frac{2qE(q^2) E(q^{16})^2}{E(q^8)},\label{eq:2-dis-E1^2}\\
	\frac{1}{E(q)E(q^3)}&=\frac{E(q^8)^2E(q^{12})^{5}}{E(q^2)^2E(q^4)E(q^6)^4E(q^{24})^2}+\frac{qE(q^4)^{5}E(q^{24})^2}{E(q^2)^4E(q^6)^2E(q^8)^2E(q^{12})},\label{eq:2-dis-1/E1E3}\\
	E(q)E(q^3)&=\frac{E(q^2)E(q^8)^2E(q^{12})^4}{E(q^4)^2E(q^{6})E(q^{24})^2}-\frac{qE(q^4)^4E(q^6)E(q^{24})^2}{E(q^2)E(q^{8})^2E(q^{12})^2},\label{eq:2-dis-E1E3}\\
	\frac{E(q)^3}{E(q^3)}&=\frac{E(q^4)^3}{E(q^{12})}-\frac{3qE(q^2)^2E(q^{12})^3}{E(q^4) E(q^6)^2},\label{eq:2-dis-E1^3/E3}\\
	\frac{E(q^3)}{E(q)^3}&=\frac{E(q^4)^6 E(q^6)^3}{E(q^2)^9 E(q^{12})^2}+\frac{3qE(q^4)^2 E(q^6) E(q^{12})^2}{E(q^2)^7},\label{eq:2-dis-E3/E1^3}\\
	\frac{E(q^3)^3}{E(q)}&=\frac{E(q^4)^3 E(q^6)^2}{E(q^2)^2 E(q^{12})}+\frac{q E(q^{12})^3}{E(q^4)}.\label{eq:2-dis-E3^3/E1}
	\end{align}
	\end{subequations}
\end{lemma}

\begin{proof}
	For \eqref{eq:2-dis-1/E1^2}, see \cite[(1.9.4)]{Hir2017}. Also, \eqref{eq:2-dis-E1^2} follows by taking $q\mapsto -q$ in \eqref{eq:2-dis-1/E1^2}. For \eqref{eq:2-dis-1/E1E3} and \eqref{eq:2-dis-E1E3}, see \cite[(30.12.3) and (30.12.1)]{Hir2017}. For \eqref{eq:2-dis-E1^3/E3}, see \cite[(22.7.3)]{Hir2017}. Also, \eqref{eq:2-dis-E3/E1^3} follows by taking $q\mapsto -q$ in \eqref{eq:2-dis-E1^3/E3}. For \eqref{eq:2-dis-E3^3/E1}, see \cite[(22.7.5)]{Hir2017}.
\end{proof}

\begin{lemma}
	We have
	\begin{subequations}
	\begin{align}
		\frac{E(q^2)^7E(q^{3})E(q^{12})}{E(q)^3E(q^4)^3E(q^6)}+\frac{2E(q^4)^3E(q^6)^2}{E(q^2)^2E(q^{12})} &= \frac{3E(q^{3})^3}{E(q)},\label{eq:13-1}\\
		\frac{E(q^2)^7E(q^{3})E(q^{12})}{E(q)^3E(q^4)^3E(q^6)}-\frac{2E(q^4)^3E(q^6)^2}{E(q^2)^2E(q^{12})} &= -\frac{E(q)^3 E(q^6)^2}{E(q^2)^2 E(q^{3})}.\label{eq:13-2}
	\end{align}
	\end{subequations}
\end{lemma}

\begin{proof}
	For \eqref{eq:13-1}, we substitute \eqref{eq:2-dis-E3/E1^3} and \eqref{eq:2-dis-E3^3/E1} into this relation and find that the two sides are equal. Also, \eqref{eq:13-2} follows by invoking \eqref{eq:2-dis-E1^3/E3} and \eqref{eq:2-dis-E3/E1^3}.
\end{proof}

\begin{lemma}
	We have
	\begin{subequations}
	\begin{align}
	G(q)G(q^4)+qH(q)H(q^4)&=\frac{E(q^2)^4}{E(q)^2E(q^4)^2},\label{eq:GH+}\\
	G(q)G(q^4)-qH(q)H(q^4)&=\frac{E(q^{10})^5}{E(q^2)E(q^5)^2E(q^{20})^2}.\label{eq:GH-}
	\end{align}
	\end{subequations}
\end{lemma}

\begin{proof}
	The two identities are among a set of 40 identities compiled by Ramanujan in his Lost Notebook \cite[pp.~236--237]{Ram1988}. See also Hirschhorn \cite[(17.4.8) and (17.4.9)]{Hir2017}.
\end{proof}

\begin{lemma}
	We have
	\begin{subequations}
	\begin{align}
	\sum_{n\ge 0}\frac{(-1)^n q^{3n^2-2n}}{(-q;q^2)_n (q^4;q^4)_n}&=\frac{(q,q^4,q^5;q^5)_\infty}{(q^2;q^2)_\infty},\label{eq:S15}\\
	\sum_{n\ge 0}\frac{(-1)^n q^{3n^2}}{(-q;q^2)_n (q^4;q^4)_n}&=\frac{(q^2,q^3,q^5;q^5)_\infty}{(q^2;q^2)_\infty},\label{eq:S19}\\
	\sum_{n\ge 0}\frac{q^{(3n^2+3n)/2}}{(q;q)_n(q;q^2)_{n+1}}&=\frac{(q^2,q^8,q^{10};q^{10})_\infty}{(q;q)_\infty},\label{eq:S44}\\
	\sum_{n\ge 0}\frac{q^{(3n^2+n)/2}}{(q;q)_n(q;q^2)_{n+1}}&=\frac{(q^4,q^6,q^{10};q^{10})_\infty}{(q;q)_\infty}.\label{eq:S46}
	\end{align}
	\end{subequations}
\end{lemma}

\begin{proof}
	These identities are all due to Rogers. See \cite[p.~330, equation (5)]{Rog1917} for \eqref{eq:S15}, \cite[p.~339, Ex.~2]{Rog1894} for \eqref{eq:S19}, \cite[p.~330, equation (2), line 2]{Rog1917} for \eqref{eq:S44}, and \cite[p.~330, equation (2), line 1]{Rog1917} for \eqref{eq:S46}. The first three identities also appear as equations (15), (19) and (44) on Slater's famous list \cite{Sla1952}. See also equations (2.5.5), (2.5.7), (2.10.4) and (2.10.6) on a more comprehensive and up-to-date list due to Mc Laughlin, Sills and Zimmer \cite{MSZ2008}.
\end{proof}

\begin{lemma}
	We have
	\begin{align}\label{eq:3phi2-1}
		{}_{3}\phi_2\left(\begin{matrix} a,b,c\\ d,e \end{matrix}; q, \frac{de}{abc}\right) = \frac{(e/a,de/(bc);q)_\infty}{(e,de/(abc);q)_\infty} {}_{3}\phi_2\left(\begin{matrix} a,d/b,d/c\\ d,de/(bc) \end{matrix}; q, \frac{e}{a}\right).
	\end{align}
\end{lemma}

\begin{proof}
	See Gasper and Rahman \cite[(III.9)]{GR2004}.
\end{proof}

\begin{lemma}
	For any integers $M,N\ge 0$,
	\begin{align}\label{eq:q-CV-spe}
	\sum_{n\ge 0}\frac{q^{n^2-(M+N)n}}{(q;q)_n(q;q)_{M-n}(q;q)_{N-n}} = \frac{q^{-MN}}{(q;q)_M(q;q)_N}.
	\end{align}
\end{lemma}

\begin{proof}
	This is the $(a,c)\mapsto (q^{-M},0)$ specialisation of the second $q$-Chu--Vandermonde summation \cite[(II.6)]{GR2004}:
	\begin{align*}
	{}_{2}\phi_{1}\left(\begin{matrix} a,q^{-N}\\ c \end{matrix}; q, q\right) = \frac{(c/a;q)_N a^N}{(c;q)_N}.
	\end{align*}
	See also \cite[Lemma D.3]{BS1998}.
\end{proof}

\begin{lemma}
	Let $\delta\in\{0,1\}$. For any integer $N$,
	\begin{align}
	\sum_{n\ge 0}\frac{q^{n^2+n(N-\delta)}}{(q;q)_n (q;q)_{n+N}}&=\frac{1+\delta q^N}{(q;q)_\infty}.\label{eq:q-Gauss-spe}
	\end{align}
\end{lemma}

\begin{proof}
	We first consider the case where $N\ge 0$. Then the two equalities for $\delta=0$ and $\delta=1$ can be derived directly by taking $b\mapsto q^N$ in an identity given by Fine \cite[(20.51)]{Fin1988}:
	\begin{align*}
		\sum_{n\ge 0} \frac{(bt)^n q^{n^2}}{(q;q)_n (bq;q)_n} = \frac{1}{(bq;q)_\infty} \sum_{n\ge 0} \frac{(t;q)_n}{(q;q)_n} (-b)^n q^{n(n+1)/2},
	\end{align*}
	followed by the specialisations of $t\mapsto 1$ and $t\mapsto q^{-1}$, respectively. When $N<0$, we make the substitution $n\mapsto n-N$ and carry out a similar calculation. We shall also notice the fact that $1/(q;q)_k = 0$ for integers $k<0$.
\end{proof}

\begin{lemma}\label{le:double-sum-rep}
	We have
	\begin{align}\label{eq:double-sum-rep}
		\sum_{m,n=-\infty}^\infty z_1^m z_2^n q^{\frac{m^2}{2}+\frac{n^2}{2}-\frac{mn}{2}}&= (-z_1^2 z_2q^{\frac{3}{2}},-z_1^{-2} z_2^{-1} q^{\frac{3}{2}},q^3;q^3)_\infty (-z_2q^{\frac{1}{2}},-z_2^{-1}q^{\frac{1}{2}},q;q)_\infty\notag\\
		&\quad+z_1q^{\frac{1}{2}}(-z_1^2 z_2q^{3},-z_1^{-2} z_2^{-1},q^3;q^3)_\infty (-z_2,-z_2^{-1}q,q;q)_\infty.
	\end{align}
\end{lemma}

\begin{proof}
	Considering the parity of $m$, we have
	\begin{align*}
		\sum_{m,n=-\infty}^\infty z_1^m z_2^n q^{\frac{m^2}{2}+\frac{n^2}{2}-\frac{mn}{2}}
		&=\sum_{u,v=-\infty}^\infty z_1^{2u} z_2^v  q^{\frac{(2u)^2}{2}+\frac{v^2}{2}-\frac{(2u)v}{2}}\\
		&\quad+\sum_{u,v=-\infty}^\infty z_1^{2u+1} z_2^v q^{\frac{(2u+1)^2}{2}+\frac{v^2}{2}-\frac{(2u+1)v}{2}}\\
		&=\sum_{u,v=-\infty}^\infty z_1^{2u}z_2^{u-(u-v)}q^{\binom{u-v}{2}+\frac{u-v}{2}+3\binom{u}{2}+\frac{3u}{2}}\\
		&\quad+z_1 q^{\frac{1}{2}}\sum_{u,v=-\infty}^\infty z_1^{2u}z_2^{u-(u-v)} q^{\binom{u-v}{2}+(u-v)+3\binom{u}{2}+3u}\\
		\text{\tiny ($w=u-v$)}&=\sum_{u=-\infty}^\infty (z_1^2 z_2)^u q^{3\binom{u}{2}+\frac{3u}{2}}\sum_{w=-\infty}^\infty z_2^{-w} q^{\binom{w}{2}+\frac{w}{2}}\\
		&\quad +z_1q^{\frac{1}{2}}\sum_{u=-\infty}^\infty (z_1^2 z_2)^u  q^{3\binom{u}{2}+3u}\sum_{w=-\infty}^\infty z_2^{-w}  q^{\binom{w}{2}+w}\\
		\text{\tiny (by \eqref{eq:JTP})}&=(-z_1^2 z_2q^{\frac{3}{2}},-z_1^{-2} z_2^{-1} q^{\frac{3}{2}},q^3;q^3)_\infty (-z_2q^{\frac{1}{2}},-z_2^{-1}q^{\frac{1}{2}},q;q)_\infty\\
		&\quad+z_1q^{\frac{1}{2}}(-z_1^2 z_2q^{3},-z_1^{-2} z_2^{-1},q^3;q^3)_\infty (-z_2,-z_2^{-1}q,q;q)_\infty,
	\end{align*}
	which is our desired result.
\end{proof}

\section{Some auxiliary series}

For the purpose of showing Theorems \ref{th:G3-general} and \ref{th:VOA-1}, we also need to establish explicit expressions for some auxiliary series.

We start with an easier one.

\begin{lemma}
	Define
	\begin{subequations}
	\begin{align}\label{eq:S_M}
	S_M:=\sum_{\substack{m,n\ge 0\\m+n=M}}\frac{x^m y^n q^{\binom{m}{2}+\binom{n}{2}}}{(q^2;q^2)_m(q^2;q^2)_n}.
	\end{align}
	Then for $M\ge 0$,
	\begin{align}\label{eq:SM-exp}
	S_M=\frac{y^M q^{\binom{M}{2}}(-xy^{-1}q^{1-M};q^2)_M}{(q^2;q^2)_M}.
	\end{align}
	\end{subequations}
\end{lemma}

\begin{proof}
Noticing that $n=M-m$, we have
\begin{align*}
S_M&=\sum_{m\ge 0}\frac{x^m y^{M-m} q^{\binom{m}{2}+\binom{M-m}{2}}}{(q^2;q^2)_m(q^2;q^2)_{M-m}}\\
&=\frac{y^M q^{\binom{M}{2}}}{(q^2;q^2)_M}\sum_{m\ge 0}(xy^{-1})^m q^{2\binom{m}{2}+(1-M)m}{M \brack m}_{q^2}\\
\text{\tiny (by \eqref{eq:q-Bin})}&=\frac{y^M q^{\binom{M}{2}}}{(q^2;q^2)_M} (-xy^{-1}q^{1-M};q^2)_M,
\end{align*}
which is exactly our desired result. For our purpose of using this relation in the proof of Theorem \ref{th:G3-general}, it would be better to restate it according to the parity of $M$:
\begin{align}
    S_{2M}&=\frac{x^M y^M q^{M^2-M} (-xy^{-1}q;q^2)_M (-x^{-1}yq;q^2)_M}{(q^2;q^2)_{2M}},\label{eq:SM-even}\\
    S_{2M+1}&=\frac{(x+y) x^M y^M q^{M^2} (-xy^{-1}q^2;q^2)_M (-x^{-1}yq^2;q^2)_M}{(q^2;q^2)_{2M+1}},\label{eq:SM-odd}
\end{align}
where $M\ge 0$.
\end{proof}

The next result is more surprising.

\begin{lemma}\label{le:T-Exp}
	Define
	\begin{align}\label{eq:T-Def}
	T_M&:=\sum_{\substack{N_1,N_2,N_3,N_4,N_5\ge 0\\N_1+N_2+2N_3+N_4+N_5=M}}\frac{z_1^{N_1+N_3+N_4}z_2^{N_2+N_3+N_5}}{(q^2;q^2)_{N_1}(q^2;q^2)_{N_2}(q^2;q^2)_{N_3}(q^2;q^2)_{N_4}(q^2;q^2)_{N_5}}\notag\\
	&\ \times q^{N_1^2+N_2^2+N_3^2+N_4^2+N_5^2+N_1N_2+N_1N_3+N_1N_5+N_2N_3+N_2N_4+N_3N_4+N_3N_5+N_4N_5-N_1-N_2}.
	\end{align}
	Then for $M\ge 0$,
	\begin{align}\label{eq:TM-exp}
	T_M=\frac{\big(z_1+z_2\big)z_2^{M-1}q^{\binom{M}{2}}(-z_1 z_2^{-1}q^{2-M};q^2)_M}{\big(1+z_1 z_2^{-1}q^M\big)(q;q)_M}.
	\end{align}

	Also, define
	\begin{align}\label{eq:T-Def-tilde}
		\widetilde{T}_M&:=\sum_{\substack{N_1,N_2,N_3,N_4,N_5\ge 0\\N_1+N_2+2N_3+N_4+N_5=M}}\frac{(-1)^{N_1+N_2}z_1^{N_1+N_3+N_4}z_2^{N_2+N_3+N_5}}{(q^2;q^2)_{N_1}(q^2;q^2)_{N_2}(q^2;q^2)_{N_3}(q^2;q^2)_{N_4}(q^2;q^2)_{N_5}}\notag\\
		&\ \times q^{N_1^2+N_2^2+N_3^2+N_4^2+N_5^2+N_1N_2+N_1N_3+N_1N_5+N_2N_3+N_2N_4+N_3N_4+N_3N_5+N_4N_5-N_1-N_2}.
	\end{align}
	Then for $M\ge 0$,
	\begin{align}\label{eq:TM-exp-tilde}
		\widetilde{T}_M=-\frac{\big(z_1-(-1)^M z_2\big)z_2^{M-1}(-q)^{\binom{M}{2}}(z_1 z_2^{-1}q^{2-M};q^2)_M}{\big(1-z_1 z_2^{-1}q^M\big)(-q;-q)_M}.
	\end{align}
\end{lemma}

\begin{proof}
We may compute the recurrences for $T_M$ and $\widetilde{T}_M$ with recourse to the \textit{Mathematica} package \texttt{qMultiSum} implemented by Riese \cite{Rie2003} of the Research Institute for Symbolic Computation (RISC) in Johannes Kepler University. This package can be downloaded at the website of RISC:
\begin{center}
	{\small\url{https://www3.risc.jku.at/research/combinat/software/ergosum/index.html}}
\end{center}

However, to reduce the amount of calculations, it is necessary to rewrite $T_M$ as
\begin{align}
T_M&=\sum_{N_1,N_2,N_3\ge 0}\frac{z_1^{N_1+N_3}z_2^{M-N_1-N_3}(-z_1z_2^{-1}q^{1+2N_2+2N_3-M};q^2)_{M-N_1-N_2-2N_3}}{(q^2;q^2)_{N_1}(q^2;q^2)_{N_2}(q^2;q^2)_{N_3}(q^2;q^2)_{M-N_1-N_2-2N_3}}\notag\\
&\times q^{M^2 - M N_1 - 2 M N_2 - 3 M N_3 + N_1^2 + 2 N_2^2 + 3 N_3^2 + 2 N_1 N_2 + 2 N_1 N_3 + 4 N_2 N_3 - N_1 - N_2}.
\end{align}
To see this, we recall $N_5=M-N_1-N_2-2N_3-N_4$ and start with the inner summation in $N_4$ by using \eqref{eq:q-Bin},
\begin{align*}
T_M&=\sum_{N_1,N_2,N_3\ge 0}\frac{z_1^{N_1+N_3}z_2^{M-N_1-N_3}}{(q^2;q^2)_{N_1}(q^2;q^2)_{N_2}(q^2;q^2)_{N_3}}\\
&\quad\times q^{M^2 - M N_1 - 2 M N_2 - 3 M N_3 + N_1^2 + 2 N_2^2 + 3 N_3^2 + 2 N_1 N_2 + 2 N_1 N_3 + 4 N_2 N_3 - N_1 - N_2}\\
&\quad\times\sum_{N_4\ge 0}\frac{z_1^{N_4}z_2^{-N_4}q^{2\binom{N_4}{2}+(1+2N_2+2N_3-M)N_4}}{(q^2;q^2)_{N_4}(q^2;q^2)_{M-N_1-N_2-2N_3-N_4}}\\
&=\sum_{N_1,N_2,N_3\ge 0}\frac{z_1^{N_1+N_3}z_2^{M-N_1-N_3}}{(q^2;q^2)_{N_1}(q^2;q^2)_{N_2}(q^2;q^2)_{N_3}(q^2;q^2)_{M-N_1-N_2-2N_3}}\\
&\quad\times q^{M^2 - M N_1 - 2 M N_2 - 3 M N_3 + N_1^2 + 2 N_2^2 + 3 N_3^2 + 2 N_1 N_2 + 2 N_1 N_3 + 4 N_2 N_3 - N_1 - N_2}\\
&\quad\times\sum_{N_4\ge 0}z_1^{N_4}z_2^{-N_4}q^{2\binom{N_4}{2}+(1+2N_2+2N_3-M)N_4} {M-N_1-N_2-2N_3 \brack N_4}_{q^2}\\
&=\sum_{N_1,N_2,N_3\ge 0}\frac{z_1^{N_1+N_3}z_2^{M-N_1-N_3}(-z_1z_2^{-1}q^{1+2N_2+2N_3-M};q^2)_{M-N_1-N_2-2N_3}}{(q^2;q^2)_{N_1}(q^2;q^2)_{N_2}(q^2;q^2)_{N_3}(q^2;q^2)_{M-N_1-N_2-2N_3}}\\
&\quad\times q^{M^2 - M N_1 - 2 M N_2 - 3 M N_3 + N_1^2 + 2 N_2^2 + 3 N_3^2 + 2 N_1 N_2 + 2 N_1 N_3 + 4 N_2 N_3 - N_1 - N_2}.
\end{align*}
Now we compute the recurrence for $T_M$ according to the parity of $M$. In particular, \eqref{eq:TM-exp} is equivalent to
\begin{subequations}
\begin{align}
T_{2M}&=\frac{z_1^M z_2^Mq^{M^2}(-z_1 z_2^{-1};q^2)_M (-z_1^{-1} z_2;q^2)_M}{(q^2;q^2)_{M}(q;q^2)_{M}},\label{eq:T-2M}\\
T_{2M+1}&=\frac{\big(z_1+z_2\big)z_1^M z_2^Mq^{M(M+1)}(-z_1 z_2^{-1} q;q^2)_M (-z_1^{-1} z_2 q;q^2)_M}{(q^2;q^2)_{M}(q;q^2)_{M+1}}.\label{eq:T-2M+1}
\end{align}
\end{subequations}
We further write $z=z_1 z_2^{-1}$ and define
\begin{align*}
	T^{\star}_{M}&:=\sum_{N_1,N_2,N_3\ge 0}\frac{z^{N_1+N_3}(-zq^{1+2N_2+2N_3-M};q^2)_{M-N_1-N_2-2N_3}}{(q^2;q^2)_{N_1}(q^2;q^2)_{N_2}(q^2;q^2)_{N_3}(q^2;q^2)_{M-N_1-N_2-2N_3}}\notag\\
	&\ \times q^{ - M N_1 - 2 M N_2 - 3 M N_3 + N_1^2 + 2 N_2^2 + 3 N_3^2 + 2 N_1 N_2 + 2 N_1 N_3 + 4 N_2 N_3 - N_1 - N_2}.
\end{align*}
Then $T_M = T^{\star}_{M} \cdot z_2^M q^{M^2}$. Therefore, it suffices to show
\begin{subequations}
\begin{align}
	T_{2M}^{\star}&=\frac{z^M q^{-3M^2}(-z;q^2)_M (-z^{-1};q^2)_M}{(q^2;q^2)_{M}(q;q^2)_{M}},\label{eq:T-2M-new}\\
	T_{2M+1}^{\star}&=\frac{\big(1+z\big)z^M q^{-3M(M+1)}(-z q;q^2)_M (-z^{-1} q;q^2)_M}{q (q^2;q^2)_{M}(q;q^2)_{M+1}}.\label{eq:T-2M+1-new}
\end{align}
\end{subequations}
We first treat \eqref{eq:T-2M-new}. To avoid ambiguity, we call $\mathscr{T}_{M}^{\star}:=T_{2M}^{\star}$ for $M\ge 0$. By invoking the \texttt{qMultiSum} package, we find that $\mathscr{T}_{M}^{\star}$ satisfies a fifth-order recurrence, so it suffices to confirm \eqref{eq:T-2M-new}, or equivalently \eqref{eq:T-2M}, for $0\le M\le 4$ and to verify that the right-hand side of \eqref{eq:T-2M-new} fulfills the same recurrence. The latter calculation can be done directly due to the fact that the ratio of the $(M+1)$-th and $M$-th terms on the right-hand side of \eqref{eq:T-2M-new} is a rational function of $q^M$, so we can effectively remove common factors when verifying the recurrence for the right-hand side of \eqref{eq:T-2M-new}. Here we will not display the concrete recurrence, which involves lengthy terms, for simplicity, but all \textit{Mathematica} codes can be found in the following webpage (notice that in the output recurrence, $\operatorname{SUM}[M]$ stands for $\mathscr{T}_{M}^{\star}$):
\begin{center}
	\url{https://github.com/shanechern/VOA}
\end{center}
We may similarly find the recurrence for $T_{2M+1}^{\star}$, as presented in the same webpage, and arrive at a proof of \eqref{eq:T-2M+1-new}.

In the same way, we rewrite \eqref{eq:T-Def-tilde} as
\begin{align*}
	\widetilde{T}_M&=\sum_{N_1,N_2,N_3\ge 0}\frac{(-1)^{N_1+N_2}z_1^{N_1+N_3}z_2^{M-N_1-N_3}(-z_1z_2^{-1}q^{1+2N_2+2N_3-M};q^2)_{M-N_1-N_2-2N_3}}{(q^2;q^2)_{N_1}(q^2;q^2)_{N_2}(q^2;q^2)_{N_3}(q^2;q^2)_{M-N_1-N_2-2N_3}}\notag\\
	&\times q^{M^2 - M N_1 - 2 M N_2 - 3 M N_3 + N_1^2 + 2 N_2^2 + 3 N_3^2 + 2 N_1 N_2 + 2 N_1 N_3 + 4 N_2 N_3 - N_1 - N_2}.
\end{align*}
Also, \eqref{eq:TM-exp-tilde} is equivalent to
\begin{subequations}
\begin{align}
	\widetilde{T}_{2M} &= \frac{z_1^M z_2^Mq^{M^2}(z_1 z_2^{-1};q^2)_M (z_1^{-1} z_2;q^2)_M}{(q^2;q^2)_{M}(-q;q^2)_{M}},\\
	\widetilde{T}_{2M+1} &= -\frac{\big(z_1+z_2\big)z_1^M z_2^Mq^{M(M+1)}(z_1 z_2^{-1} q;q^2)_M (z_1^{-1} z_2 q;q^2)_M}{(q^2;q^2)_{M}(-q;q^2)_{M+1}}.
\end{align}
\end{subequations}
The remaining computer-assisted calculations are analogous.
\end{proof}

\section{More identities of Rogers--Ramanujan type}

In Section \ref{sec:Level-2}, we have already witnessed several Rogers--Ramanujan type identities as specialisations of the transforms in Section \ref{sec:transform}. Since such instances are of interest within the community of $q$-theorists, we will record more in this appendix.

First, we reload the two identities from Theorem \ref{th:G3-general}, namely, \eqref{eq:triple-AG-1-2} and \eqref{eq:triple-AG-2-2}.

\begin{theorem}
Let $\mathbf{G}_3$ and $\mathbf{N}$ be as in \eqref{eq:G3-def} and \eqref{eq:G3-N}, respectively. Then,
\begin{subequations}
\begin{align}
	&\sum_{N_1,N_2,N_3\ge 0} \frac{q^{\mathbf{N}^{\mathsf{T}}\cdot\mathbf{G}_3\cdot\mathbf{N}}}{(q^2;q^2)_{N_1}(q^2;q^2)_{N_2}(q^2;q^2)_{N_3}}\notag\\
	&\qquad=\frac{(q;q)_\infty(q^2,q^{3},q^{5};q^{5})_\infty}{(q^2;q^2)^2_\infty}+\frac{4q(q^4;q^4)_\infty(q^4,q^{16},q^{20};q^{20})_\infty}{(q^2;q^2)^2_\infty},\\
	&\sum_{N_1,N_2,N_3\ge 0} \frac{q^{\mathbf{N}^{\mathsf{T}}\cdot\mathbf{G}_3\cdot\mathbf{N}-N_2-N_3}}{(q^2;q^2)_{N_1}(q^2;q^2)_{N_2}(q^2;q^2)_{N_3}}\notag\\
	&\qquad=-\frac{(q;q)_\infty(q,q^{4},q^{5};q^{5})_\infty}{(q^2;q^2)^2_\infty}+\frac{4(q^4;q^4)_\infty(q^8,q^{12},q^{20};q^{20})_\infty}{(q^2;q^2)^2_\infty}.
\end{align}
\end{subequations}
\end{theorem}

We then take a look at the sextuple summations in Theorem \ref{th:VOA-z}. First, we should keep in mind that the right-hand side of \eqref{eq:VOA-z-1} is always possible to be represented as a linear combination of products of theta series by invoking Lemma \ref{le:double-sum-rep}. However, some cases, including \eqref{eq:VOA-1-1} and \eqref{eq:VOA-1-2} are more interesting.

\begin{theorem}\label{th:G4-1}
Let $\mathbf{G}_4$ and $\mathbf{N}$ be as in \eqref{eq:G4-def} and \eqref{eq:G4-N}, respectively. Then,
\begin{subequations}
\begin{align}
	&\sum_{N_1,N_2,N_3,N_4,N_5,N_6\ge 0} \frac{q^{\mathbf{N}^{\mathsf{T}}\cdot \mathbf{G}_4\cdot \mathbf{N}}}{(q^2;q^2)_{N_1}(q^2;q^2)_{N_2}(q^2;q^2)_{N_3}(q^2;q^2)_{N_4}(q^2;q^2)_{N_5}(q^2;q^2)_{N_6}}\notag\\
	&\qquad=\frac{(q^2;q^2)^3_\infty(q^6;q^6)^5_\infty}{(q;q)^2_\infty(q^3;q^3)^2_\infty(q^4;q^4)^2_\infty(q^{12};q^{12})^2_\infty}+\frac{4q(q^4;q^4)^2_\infty(q^{12};q^{12})^2_\infty}{(q^2;q^2)^3_\infty(q^6;q^6)_\infty},\\
	&\sum_{N_1,N_2,N_3,N_4,N_5,N_6\ge 0} \frac{(-1)^{N_1+N_2+N_4+N_5}q^{\mathbf{N}^{\mathsf{T}}\cdot \mathbf{G}
	_4\cdot \mathbf{N}}}{(q^2;q^2)_{N_1}(q^2;q^2)_{N_2}(q^2;q^2)_{N_3}(q^2;q^2)_{N_4}(q^2;q^2)_{N_5}(q^2;q^2)_{N_6}}\notag\\
	&\qquad=\frac{(q;q)^2_\infty(q^3;q^3)^2_\infty}{(q^2;q^2)^3_\infty(q^6;q^6)_\infty},\label{eq:VOA-1-1-T}\\
	&\sum_{N_1,N_2,N_3,N_4,N_5,N_6\ge 0} \frac{q^{\mathbf{N}^{\mathsf{T}}\cdot \mathbf{G}_4\cdot \mathbf{N}-(N_1+N_3+N_4+N_6)}}{(q^2;q^2)_{N_1}(q^2;q^2)_{N_2}(q^2;q^2)_{N_3}(q^2;q^2)_{N_4}(q^2;q^2)_{N_5}(q^2;q^2)_{N_6}}\notag\\
	&\qquad=\frac{6(q^3;q^3)^3_\infty}{(q;q)_\infty(q^2;q^2)^2_\infty},\\
	&\sum_{N_1,N_2,N_3,N_4,N_5,N_6\ge 0} \frac{(-1)^{N_1+N_3+N_4+N_6}q^{\mathbf{N}^{\mathsf{T}}\cdot \mathbf{G}_4\cdot \mathbf{N}-(N_1+N_3+N_4+N_6)}}{(q^2;q^2)_{N_1}(q^2;q^2)_{N_2}(q^2;q^2)_{N_3}(q^2;q^2)_{N_4}(q^2;q^2)_{N_5}(q^2;q^2)_{N_6}}\notag\\
	&\qquad=-\frac{2(q;q)^3_\infty (q^6;q^6)^2_\infty}{(q^2;q^2)^4_\infty (q^3;q^3)_\infty},\label{eq:VOA-1-2-T}\\
	&\sum_{N_1,N_2,N_3,N_4,N_5,N_6\ge 0} \frac{(-1)^{N_1+N_2+N_4+N_5}q^{\mathbf{N}^{\mathsf{T}}\cdot \mathbf{G}_4\cdot \mathbf{N}-(N_1+N_3+N_4+N_6)}}{(q^2;q^2)_{N_1}(q^2;q^2)_{N_2}(q^2;q^2)_{N_3}(q^2;q^2)_{N_4}(q^2;q^2)_{N_5}(q^2;q^2)_{N_6}}\notag\\
	&\qquad=\frac{2(q;q)^3_\infty (q^6;q^6)^2_\infty}{(q^2;q^2)^4_\infty (q^3;q^3)_\infty}.\label{eq:VOA-1-2-T-T}
	\end{align}
\end{subequations}
\end{theorem}

\begin{proof}
    We only need to establish \eqref{eq:VOA-1-1-T}, \eqref{eq:VOA-1-2-T} and \eqref{eq:VOA-1-2-T-T}. For \eqref{eq:VOA-1-1-T}, we use \eqref{eq:VOA-z-1} with $\delta_1=\delta_2=0$ and $z_1=z_2=-1$ and then apply \eqref{eq:double-sum-rep} with $q$ replaced by $q^2$. For \eqref{eq:VOA-1-2-T}, we use \eqref{eq:VOA-z-1} with $\delta_1=1$, $\delta_2=0$ and $(q,z_1,z_2)\mapsto (q^2,-q^2,1)$, and then take advantage of \eqref{eq:double-sum-rep} followed by an application of \eqref{eq:13-2}. Finally, \eqref{eq:VOA-1-2-T-T} is the $\delta_1=1$, $\delta_2=0$ and $(q,z_1,z_2)\mapsto (q^2,-q^2,-1)$ case of \eqref{eq:VOA-z-1}.
\end{proof}

Finally, we notice that the sextuple summations in Theorem \ref{th:VOA-1} reduce to linear combinations of infinite products if either of $z_1$ and $z_2$ is $\pm 1$. Due to the symmetry, we only record the specialisation of $z_2\in\{\pm 1\}$.

\begin{theorem}\label{th:G4-2}
Let $\mathbf{G}_4$ and $\mathbf{N}$ be as in \eqref{eq:G4-def} and \eqref{eq:G4-N}, respectively. Then,
\begin{subequations}
\begin{align}
		&\sum_{N_1,N_2,N_3,N_4,N_5,N_6\ge 0} \frac{z^{N_1+N_3+N_4+N_6} q^{\mathbf{N}^{\mathsf{T}}\cdot \mathbf{G}_4\cdot \mathbf{N}-(N_1+N_2+N_6)}}{(q^2;q^2)_{N_1}(q^2;q^2)_{N_2}(q^2;q^2)_{N_3}(q^2;q^2)_{N_4}(q^2;q^2)_{N_5}(q^2;q^2)_{N_6}}\notag\\
		&\qquad=\frac{2(- z;q)_\infty}{(q;q^2)_\infty},\label{eq:G4-2-sp1}\\
		&\sum_{N_1,N_2,N_3,N_4,N_5,N_6\ge 0} \frac{(-1)^{N_2+N_3+N_5+N_6}z^{N_1+N_3+N_4+N_6} q^{\mathbf{N}^{\mathsf{T}}\cdot \mathbf{G}_4\cdot \mathbf{N}-(N_1+N_2+N_6)}}{(q^2;q^2)_{N_1}(q^2;q^2)_{N_2}(q^2;q^2)_{N_3}(q^2;q^2)_{N_4}(q^2;q^2)_{N_5}(q^2;q^2)_{N_6}}\notag\\
		&\qquad=0,\label{eq:G4-2-sp2}\\
		&\sum_{N_1,N_2,N_3,N_4,N_5,N_6\ge 0} \frac{(-1)^{N_1+N_2+N_6}z^{N_1+N_3+N_4+N_6} q^{\mathbf{N}^{\mathsf{T}}\cdot \mathbf{G}_4\cdot \mathbf{N}-(N_1+N_2+N_6)}}{(q^2;q^2)_{N_1}(q^2;q^2)_{N_2}(q^2;q^2)_{N_3}(q^2;q^2)_{N_4}(q^2;q^2)_{N_5}(q^2;q^2)_{N_6}}\notag\\
		&\qquad=\frac{(z;-q)_\infty}{(-q;q^2)_\infty}-\frac{(-z;-q)_\infty}{(-q;q^2)_\infty},\label{eq:G4-2-sp3}\\
		&\sum_{N_1,N_2,N_3,N_4,N_5,N_6\ge 0} \frac{(-1)^{N_1+N_3+N_5}z^{N_1+N_3+N_4+N_6} q^{\mathbf{N}^{\mathsf{T}}\cdot \mathbf{G}_4\cdot \mathbf{N}-(N_1+N_2+N_6)}}{(q^2;q^2)_{N_1}(q^2;q^2)_{N_2}(q^2;q^2)_{N_3}(q^2;q^2)_{N_4}(q^2;q^2)_{N_5}(q^2;q^2)_{N_6}}\notag\\
		&\qquad=\frac{(z;-q)_\infty}{(-q;q^2)_\infty}+\frac{(-z;-q)_\infty}{(-q;q^2)_\infty}.\label{eq:G4-2-sp4}
	\end{align}
\end{subequations}
\end{theorem}

\begin{proof}
In \eqref{eq:VOA-1-3} and \eqref{eq:VOA-1-3-T-1}, we take $(z_1,z_2)\mapsto (z,1)$, then \eqref{eq:G4-2-sp1} and \eqref{eq:G4-2-sp3} follow, respectively. For \eqref{eq:G4-2-sp2} and \eqref{eq:G4-2-sp4}, the only difference is that we change $z_2$ to $-1$.
\end{proof}


\end{document}